\documentclass[a4paper,12pt]{amsart}
\usepackage{latexsym}
\usepackage{amssymb}
\usepackage{amsfonts}
\usepackage{epsfig}
\usepackage{amsmath}
\usepackage{amscd}
\usepackage[active]{srcltx}
\usepackage{psfrag}
\usepackage{array}
\usepackage{natbib}
\usepackage{tikz}
\usepackage[all]{xy}
\usepackage{xcolor}
\usepackage{colortbl}
\usepackage{graphicx}
\usepackage{color}
\usepackage{url}
\usepackage{amsmath}
\usepackage{multirow,bigdelim}
\usepackage{enumerate}
\usepackage{booktabs}
\usepackage{verbatim}
\usepackage{lscape}
\usepackage[ruled,vlined]{algorithm2e}
\usepackage{float}
\usepackage{yhmath}
\usepackage{soul}
\usepackage{anysize}
\usepackage[font=small,labelfont=bf]{caption}

\newtheorem{defn0}{Definition}[section]
\newtheorem{prop0}[defn0]{Proposition}
\newtheorem{thm0}[defn0]{Theorem}
\newtheorem{lemma0}[defn0]{Lemma}
\newtheorem{corollary0}[defn0]{Corollary}
\newtheorem{example0}[defn0]{Example}
\newtheorem{conjecture0}[defn0]{Conjecture}
\newtheorem{notation0}[defn0]{Notation}
\newtheorem{remark0}[defn0]{Remark}
\newtheorem{assumption0}[defn0]{$d$-claw tree hypothesis}
\newtheorem{problem0}[defn0]{Problem}

\newenvironment{defi}{\begin{defn0} \rm}{\end{defn0}}
\newenvironment{prop}{\begin{prop0}}{\end{prop0}}
\newenvironment{thm}{\begin{thm0}}{\end{thm0}}
\newenvironment{lema}{\begin{lemma0}}{\end{lemma0}}
\newenvironment{cor}{\begin{corollary0}}{\end{corollary0}}
\newenvironment{example}{\begin{example0} \rm}{\end{example0}}
\newenvironment{rk}{\begin{remark0} \rm}{\end{remark0}}
\newenvironment{conjecture}{\begin{conjecture0}\rm}{\end{conjecture0}}
\newenvironment{notation}{\begin{notation0} \rm}{\end{notation0}}

\newenvironment{pb}{\begin{problem0} \rm}{\end{problem0}}

\DeclareMathOperator*{\argmin}{arg\,min}
\DeclareMathOperator{\proj}{proj}

\newcommand{\RR}{\mathbb{R}}

\newcommand{\der}[2]{\partial_{#2} (#1)}
\newcommand{\dder}[3]{\partial_{#2} \, #1 (#3)}
\newcommand{\xstar}{\mathbf{x}^*}
\newcommand{\mstar}{\mathrm{m}(k)}

\newcommand{\CC}{\mathbb{C}}
\newcommand{\TT}{\mathcal{T}}

\newcommand{\VV}{\mathcal{V}}
\newcommand{\gkm}{\gamma(k,m)}
\newcommand{\uu}{\mathbf{u}}
\newcommand{\Hess}{\mathbf{H}}

\newcommand{\place}[2]{\overset{\substack{#1\\\smile}}{#2}}
\newcommand{\placedown}[2]{\underset{\substack{\frown\vspace{-0.07cm}\\ #1}}{#2}}

\makeatletter
\newsavebox\myboxA
\newsavebox\myboxB
\newlength\mylenA

\newcommand*\overbar[2][0.75]{%
	\sbox{\myboxA}{$\m@th#2$}%
	\setbox\myboxB\null
	\ht\myboxB=\ht\myboxA%
	\dp\myboxB=\dp\myboxA%
	\wd\myboxB=#1\wd\myboxA
	\sbox\myboxB{$\m@th\overline{\copy\myboxB}$}
	\setlength\mylenA{\the\wd\myboxA}
	\addtolength\mylenA{-\the\wd\myboxB}%
	\ifdim\wd\myboxB<\wd\myboxA%
	\rlap{\hskip 0.5\mylenA\usebox\myboxB}{\usebox\myboxA}%
	\else
	\hskip -0.5\mylenA\rlap{\usebox\myboxA}{\hskip 0.5\mylenA\usebox\myboxB}%
	\fi}
\makeatother

\marginsize{2.5cm}{2.5cm}{1.5cm}{3cm}

\begin{document}
	\title{Distance to the stochastic part of phylogenetic varieties}
	\author{Marta Casanellas}
	\author{Jes\'us Fern\'andez-S\'anchez}
	\author{Marina Garrote-L\'opez}

\begin{abstract}
Modelling the substitution of nucleotides along a phylogenetic tree is
usually done by a hidden Markov process. This allows to define a distribution
of characters at the leaves of the trees and one might be able to obtain
polynomial relationships among the probabilities of different characters.
The study of these polynomials and the geometry of the algebraic varieties
defined by them can be used to reconstruct phylogenetic trees. However,
not all points in these algebraic varieties have biological sense. In this
paper, we explore the extent to which adding semi-algebraic conditions
arising from the restriction to parameters with statistical meaning can
improve existing methods of phylogenetic reconstruction. To this end, our
aim is to compute the distance of data points to algebraic varieties and
to the stochastic part of these varieties. Computing these distances involves
optimization by nonlinear programming algorithms. We use analytical methods
to find some of these distances for quartet trees evolving under the Kimura
3-parameter or the Jukes-Cantor models. Numerical algebraic geometry and
computational algebra play also a fundamental role in this paper.
\end{abstract}

\maketitle


\normalsize

\section{Introduction}

Within the new century, algebraic tools have started to be successfully
applied to some problems of phylogenetic reconstruction, see for example
\citet{allmandegnanrhodes2013}, \citet{chifmankubatko2015} and \citet{allmankubatkorhodes}.
The main goal of phylogenetic reconstruction is to estimate the
\emph{phylogenetic tree} that best explains the evolution of living species
using solely information of their genome. To this end, one usually considers
evolutionary models of molecular substitution and assume that DNA sequences
evolve according to these models by a Markov process on a tree. Some of
the most used models are \emph{nucleotide substitution models} (e.g. \cite{kimura1981} or \cite{JC69} models), which
are specified by a $4\times 4$ transition matrix associated to each edge
of the tree and a distribution of nucleotides at the root. Then, the distribution
of possible nucleotide sequences at the leaves of the tree (representing
the living species) can be computed as an algebraic expression in terms
of the parameters of the model (the entries of the substitution matrices
and the distribution at the root). This allows the use of algebraic tools
for phylogenetic reconstruction purposes.

When reconstructing the \emph{tree topology} (i.e., the shape of the tree
taking into account the names of the species at the leaves), the main tools
that have been used come either from rank conditions on matrices arising
from a certain rearrangement of the distribution of nucleotides at the
leaves \citep{SVDquartets,chifmankubatko2015,casfer2016}, or from phylogenetic
invariants \citep{lake1987,casanellas2007}. These tools use the fact that
the set of possible distributions satisfies certain \emph{algebraic} constraints,
but do not specifically use the condition that one is dealing with discrete
\emph{distributions} that arise from \emph{stochastic} matrices at the edges
of the tree (i.e. with positive entries and rows summing to one). These
extra conditions lead to \emph{semi-algebraic} constraints which have been
specified for certain models by \cite{AllmanSemialg} (for the general Markov
model), \cite{matsen2009} (for the Kimura 3-parameter model) and by
\cite{zwierniksmith} and \cite{klaere2012} for the 2-state case ($2\times 2$ transition
matrices). Combining algebraic and semi-algebraic conditions to develop
a tool for reconstructing the tree topology is not an easy task and, as
far as we are aware, both tools have only been used together in
\cite{Kosta2019} for the simple case of 2 states.

As a starting point of topology reconstruction problems, it is natural
to use trees on four species (called 1, 2, 3, 4 for example). In this case,
there are three possible (unrooted and fully resolved) phylogenetic trees,
$13|24$, $13|24$, and $14|23$ (see Fig.~\ref{Fig:3_top}). Then a distribution
of nucleotides for this set of species is a vector
$P\in \mathbb{R}^{4^{4}}$ whose entries are non-negative and sum to one.
The set of distributions arising from a Markov process on any of these
trees $T$ (for a given substitution model) defines an algebraic variety
$\mathcal{V}_{T}$ (see Section~\ref{sec:PhyloVars}). The three
\emph{phylogenetic varieties} $\mathcal{V}_{12|34}$,
$\mathcal{V}_{13|24}$, $\mathcal{V}_{14|23}$ are different and the topology
reconstruction problem for a given distribution
$P\in \mathbb{R}^{4^{4}}$ is, briefly, deciding to which of these three
varieties $P$ is closest (for a certain distance or for another specified
optimization problem such as likelihood estimation). The algebraic tools
related to rank conditions mentioned above attempt to estimate these Euclidean
distances, for example.

\begin{figure}
\footnotesize 
\includegraphics[scale=0.3]{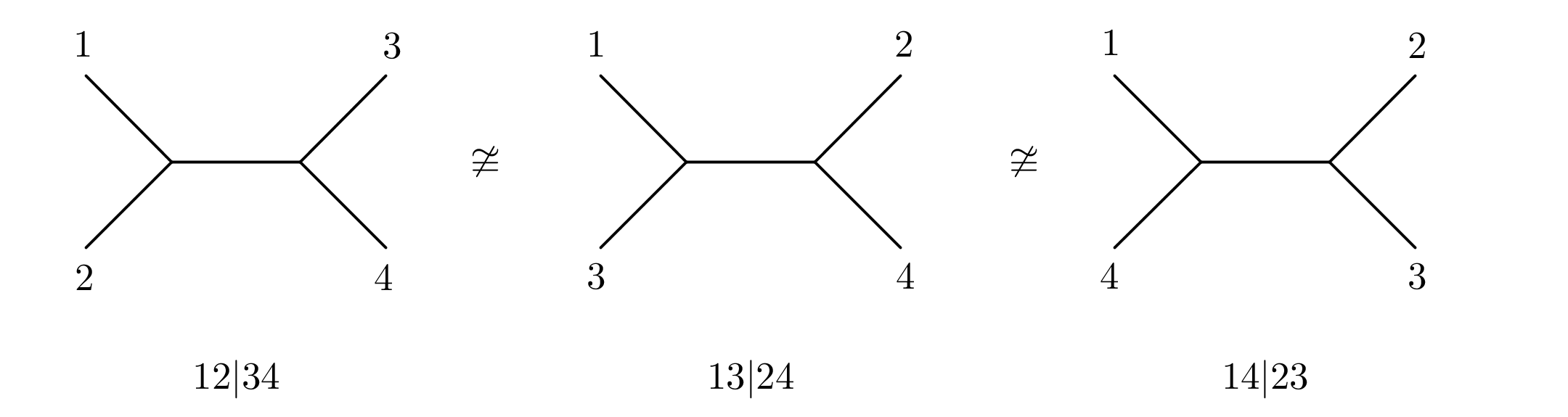}
\caption{The three unrooted (fully resolved) phylogenetic trees on $4$ leaves:
$12|34$ (left), $13|24$ (middle) and $14|23$ (right).}
\label{Fig:3_top}
\end{figure}

If we assume that $P$ should be close to a distribution that has arisen
from stochastic parameters on one of these trees, then one should consider
only the \emph{stochastic part} of these varieties,
$\mathcal{V}_{12|34}^{+}$, $\mathcal{V}_{13|24}^{+}$,
$\mathcal{V}_{14|23}^{+}$ (which we call the
\emph{stochastic phylogenetic regions}). The main questions that motivated
the study presented here are:
\begin{itemize}
\item[-]
\emph{Could semi-algebraic tools add some insight to the already existent
algebraic tools?}
\item[-]
\emph{Do semi-algebraic conditions support the same tree $T$ whose algebraic
variety $\mathcal{V}_{T}$ is closest to the data point?}
\end{itemize}

In terms of the Euclidean distance and trees of four species, we make the
explicit following question:
\begin{itemize}
\item[(*)]
\emph{If $P\in \mathbb{R}^{4^{4}}$ is a distribution satisfying
$d(P,\mathcal{V}_{12|34}) <min\{d(P,\mathcal{V}_{13|24}),d(P,
\mathcal{V}_{14|23})\}$, would it be possible that
$d(P,\mathcal{V}_{12|34}^{+}) > \min \{d(P,\mathcal{V}_{13|24}^{+}),d(P,
\mathcal{V}_{14|23}^{+})\}$?}
\end{itemize}

We address this problem for special cases of interest in phylogenetics:
short branches at the external edges (see section~\ref{sec4}) and long branch attraction
(in section~\ref{sec:lba}). The length of a branch in a phylogenetic tree is understood
as the expected number of substitutions of nucleotides per site along the
corresponding edge; both cases, short and long branches, usually lead to
confusing results in phylogenetic reconstruction (particularly in relation
to the long branch attraction problem, see section~\ref{sec:lba}). In the first case
we are able to deal with the Kimura 3-parameter model and in the second
case we have to restrict to the more simple Jukes-Cantor (JC69) model.
The reason for this restriction is that the computations get more involved
in the second case and we have to use computational algebra techniques
(for which is crucial to decrease the number of variables of the problem).
To this end, in section~\ref{sec5} we introduce an algorithm that computes the distance
of a point to the stochastic phylogenetic regions in the JC69 case; this
algorithm makes explicit use of the Euclidean distance degree
\citep{draisma2015} of the phylogenetic varieties.

We find that in the first framework (short external branches), restricting
to the stochastic part does not make any difference, that is, Question
1 has a negative answer in this case (see Theorem~\ref{thm:seb}). However,
in the long branch attraction framework, considering the stochastic part
of phylogenetic varieties might be of interest, specially if the data points
are close to the intersection of the three varieties, see Theorem~\ref{thm_lba}.
In particular, the answer to Question 1 is positive for
data close to the long branch attraction problem under the JC69 model.
In section~\ref{sec:simulations} we provide results on simulated data that support these findings
and also show a positive answer to Question 1 for balanced trees.

Summing up, incorporating the semi-algebraic conditions to the problem
of phylogenetic reconstruction seems important when the data are close
to the intersection of the three phylogenetic varieties. This is the case
where phylogenetic reconstruction methods tend to confuse the trees. On
the contrary, on data points which are far from the intersection (in the
short branches case of section~\ref{sec4} for example), it does not seem necessary
to incorporate these semi-algebraic tools. This is the reason why incorporating
these tools into phylogenetic reconstruction methods might be extremely
difficult.

In this paper we consider only the Euclidean distance. One reason to do
so is that the initial algebraic tools based on rank conditions were dealing
with it, but another motivation is that the algebraic expression of the
Euclidean distance permits the use of algebraic tools to derive analytical
results and the use of numerical algebraic geometry to get global minima.
On the other hand, the use of other measures such as Hellinger distance
or maximum likelihood, would not allow the use of the Fourier transform
for the evolutionary models we use here, which significantly simplifies
the computations in our case.

The organization of the paper is as follows. In section \ref{sec:Prel}, we introduce
the concepts on nucleotide substitution models and phylogenetic varieties
that we will use later on. Then in section~\ref{sec:matrix} we prove some technical results
regarding the closest stochastic matrix to a given matrix. In section~\ref{sec4}
we consider the case of short external branches for the Kimura 3-parameter
model and obtain the results analytically. In section~\ref{sec5} we introduce the
computational approach that we use in order to compute the distance to
the stochastic phylogenetic regions. The results for the long branch attraction
case are expanded in section \ref{sec:lba} and in section~\ref{sec:simulations} we provide results on simulated
data that illustrate our findings. The Appendix collects all technical
proofs needed in section \ref{sec:lba}.

\section{Preliminaries}\label{sec:Prel}

\subsection{Phylogenetic varieties}
\label{sec:PhyloVars}

We refer the reader to the work by \cite{AllmanRhodeschapter4} for a good general overview of phylogenetic algebraic
geometry. Here we briefly introduce the basic concepts that will be needed
later. Let $T$ be a \emph{quartet} tree topology, that is, an (unrooted)
trivalent phylogenetic tree with its leaves labelled by
$\{1,2,3,4\}$ (i.e. $T$ is a connected acyclic graph whose interior nodes
have degree $3$ and whose leaves, of degree 1, are in correspondence with
$\{1,2,3,4\}$), see Fig.~\ref{Fig:3_top}. Using the notation introduced
in Fig.~\ref{Fig:3_top}, $T$ belongs to the set
$\mathcal{T}=\{12|34, 13|24, 14|23\}$. We choose an internal vertex as
the root $r$ of $T$, which induces an orientation on the set of edges
$E(T)$. Suppose the Markovian evolutionary process on that tree follows
a nucleotide substitution model $\mathcal{M}$: associate a random variable
taking values on $\Sigma :=\{\mathtt{A},\mathtt{C},\mathtt{G},\mathtt{T }\}$ at each node of
the tree, and consider as parameters a distribution
$\pi = (\pi _{\mathtt{A}}, \pi _{\mathtt{C}}, \pi _{\mathtt{G}}, \pi _{\mathtt{T}})$ at the
root, $\sum _{i}\pi _{i} = 1$, and a $4\times 4$ transition matrix
$M_{e}$ at each (oriented) edge $e$ of $T$. The transition matrices are
\textit{stochastic} (or \textit{Markov}) matrices, that is, all its entries
are non-negative and its rows sum up to $1$. A vector is
\textit{stochastic} if all its entries are nonnegative and sum up to
$1$.

If $T\in \mathcal{T}$ and $S$ is the set of stochastic parameters described
above, we denote by $\psi _{T}$ the following (parametrization) map:
\begin{align*}
\psi _{T}: S\subset [0,1]^{\ell } &\rightarrow \mathbb{R}^{4^{4}}
\\
\{\pi ,\{M_{e}\}_{e\in E(T)}\} &\mapsto P=(p_{\mathtt{AAAA}}, p_{\mathtt{AAAC}},
\ldots , p_{\mathtt{TTTG}}, p_{\mathtt{TTTT}})
\end{align*}
which maps each set of parameters of the model
$\{\pi ,\{M_{e}\}_{e\in E(T)}\}\in S$ to the joint distribution of characters
at the leaves of $T$ given by a hidden Markov process on $T$ governed by
these parameters. The entries $p_{x_{1},\ldots ,x_{4}}$ of the joint distribution
can be expressed in terms of the entries of the substitution matrices.
We adopt the following notation: trees are rooted at the interior node
neighbour to leaf 1, for $i=1,\ldots , 4$, $M_{i}$ is the transition matrix
on the edge ending at leaf $i$, and $M_{5}$ is the transition matrix at
the interior edge. For example, for the tree $12|34$ rooted at the leftmost
internal edge with transition matrices as in Fig.~\ref{Fig:tree} we have
\begin{equation*}
p_{x_{1},x_{2},x_{3},x_{4}}=\sum _{x_{r},x_{s}\in \Sigma }\pi _{x_{r}} M_{1}(x_{r},x_{1})M_{2}(x_{r},x_{2})M_{5}(x_{r},x_{s})M_{3}(x_{s},x_{3})M_{4}(x_{s},x_{5}).
\end{equation*}

\begin{figure}
\includegraphics[scale=0.5]{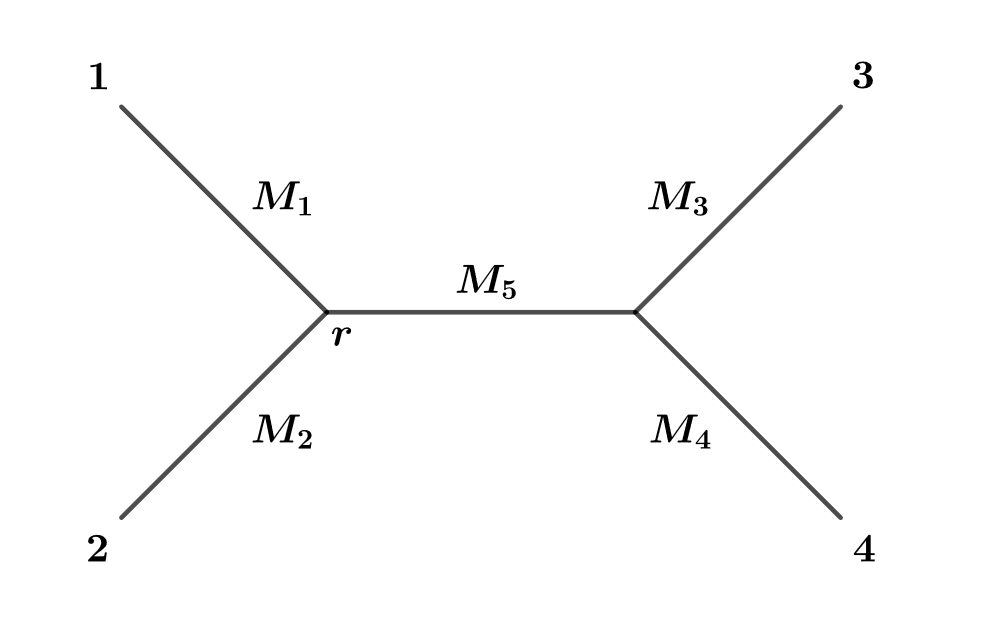}
\caption{Tree $12|34$ with transition matrices $M_{1}$, $M_{2}$, $M_{3}$,
$M_{4}$ and $M_{5}$.}
\label{Fig:tree}
\end{figure}

We write $\mathcal{V}_{T}^{+}$ for the image of this map, that is, the
space of all the distributions arising from stochastic parameters,
\begin{equation*}
\mathcal{V}_{T}^{+}=\{P \in \mathcal{V}_{T}\ |\ P = \psi _{T}(s)
\mbox{ and } s\in S\}.
\end{equation*}
We call this set the stochastic phylogenetic region.

Since $\psi _{T}$ is a polynomial map, it can be extended to
$\mathbb{R}^{\ell }$. That is, we can consider not only nonnegative entries
in $\pi $ and $M_{e}$, but we always assume that the rows of the matrices
$M_{e}$ and the vector $\pi $, sum up to $1$. Define the
\textit{phylogenetic variety} associated with $\mathcal{T}$ as the Zariski
closure of $\psi _{T}(\mathbb{R}^{l})$,
\begin{equation*}
\mathcal{V}_{T} = \overline{\psi _{T}(\mathbb{R}^{l})}.
\end{equation*}

This variety contains all joint distributions that arise from stochastic
parameters on the tree $T$, but not every point in this variety is of this
type.

Although the choice of a root was necessary to define the map
$\psi _{T}$, the phylogenetic variety and the stochastic region do not
depend on it \citep{allman2003}. Throughout the paper we consider the Euclidean
distance between points, even if we do not specify it.

\subsection{Kimura and Jukes-Cantor models}
\label{sec2.2}

In this paper we focus on phylogenetic $4$-leaf trees evolving under the
\textit{Jukes-Cantor} model ($JC69$ for short, see \cite{JC69}) and also the
\textit{$3$-parameter Kimura} model ($K81$ for short, see \cite{kimura1981}).
The JC69 model is a highly structured model that assumes equal mutation
probabilities and the K81 takes into account the classification of nucleotides
as purines/pyrimidines and the probabilities of substitution between and
within these groups; both models assume the uniform distribution at the
root, $\pi = (\frac{1}{4}, \frac{1}{4}, \frac{1}{4}, \frac{1}{4})$.

\begin{defi}
A $4\times 4$ matrix $M$ is a \emph{K81 matrix} if it is of the form
%
\begin{equation}
\label{eq:K81}
\footnotesize
M = \left (
\begin{array}{c@{\quad }c@{\quad }c@{\quad }c}
a & b & c & d
\\
b & a & d & c
\\
c & d & a & b
\\
d & c & b & a
\\
\end{array}
\right ),
\end{equation}
for some $a,b,c,d\in \mathbb{R}$ summing to 1, $a+b+c+d = 1$. If
$b=c=d$, then we say that $M$ is a \emph{JC69 matrix}.
\end{defi}

Note that these matrices only have an interpretation as transition matrices
of a Markov process if they have nonnegative entries; in this case we talk
about \emph{stochastic K81 matrices} or
\emph{stochastic JC69 matrices}.

\begin{lema}%
\label{lema:eigenK81}%
\emph{\citep{ARbook}} If $M$ is a $K81$ matrix as \eqref{eq:K81}, then
it diagonalizes with eigenvalues $m_{\mathtt{A}} = a+b+c+d=1$,
$m_{\mathtt{C}} = a+b-c-d$, $m_{\mathtt{G}} = a-b+c-d$ and
$m_{\mathtt{T}} = a-b-c+d$ and respective eigenvectors
$\bar{\mathtt{A}}=(1,1,1,1)^{t}$, $\bar{\mathtt{C}}=(1,1,-1,-1)^{t}$,
$\bar{\mathtt{G}}=(1,-1,1,-1)^{t}$ and $\bar{\mathtt{T}}=(1,-1,-1,1)^{t}$. In particular,
the eigenvalues of a $JC69$ matrix are $m_{\mathtt{A}} = 1$ and
$m_{\mathtt{C}} = m_{\mathtt{G}} = m_{\mathtt{T}} = 1-4b$.
\end{lema}

\subsection{Fourier coordinates and Fourier parameters}
\label{sec2.3}

Let $M$ be a $K81$ matrix and write
$m_{\mathtt{A}}, m_{\mathtt{C}}, m_{\mathtt{G}}, m_{\mathtt{T}}$ and
$\bar{\mathtt{A}}, \bar{\mathtt{C}}, \bar{\mathtt{G}}, \bar{\mathtt{T}}$ for the eigenvalues
and eigenvectors of $M$, respectively. The basis of eigenvectors will be
denoted by
$\overbar {\Sigma }=\{\bar{\mathtt{A}}, \bar{\mathtt{C}}, \bar{\mathtt{G}},
\bar{\mathtt{T}}\}$ and is called the \emph{Fourier basis}. Because of Lemma~\ref{lema:eigenK81}, we have
\begin{eqnarray*}
\overbar {M} = H^{-1}\cdot M \cdot H,
\end{eqnarray*}
where
$\overbar {M} = diag(m_{\mathtt{A}}, m_{\mathtt{C}}, m_{\mathtt{G}}, m_{\mathtt{T}})$ and
\begin{equation*}
\footnotesize
H=\left (
\begin{array}{c@{\quad }c@{\quad }c@{\quad }c}
1 & 1 & 1 & 1
\\
1 & 1 & -1 & -1
\\
1 & -1 & 1 & -1
\\
1 & -1 & -1 & 1
\\
\end{array}
\right )
\end{equation*}
is the matrix of change of basis from $\overbar {\Sigma }$ to
$\Sigma $. Notice that $H^{-1} = \frac{1}{4}H^{t} = \frac{1}{4}H$. The
eigenvalues
$m^{i}_{\mathtt{A}}, m^{i}_{\mathtt{C}}, m^{i}_{\mathtt{G}}, m^{i}_{\mathtt{T}}$ of
$M_{i}$ will be called \emph{Fourier parameters}.

The vectors
$P = (p_{\mathtt{AAAA}}, p_{\mathtt{AAAC}}, \ldots , p_{\mathtt{TTTG}}, p_{\mathtt{TTTT}})
\in \mathbb{R}^{4^{4}}$ considered in section~\ref{sec:PhyloVars} can be thought of as
$4\times 4\times 4\times 4$ tensors in
$\left (\mathbb{R}^{4}\right )^{\otimes 4}$: if we call
$\Sigma =\{\mathtt{A},\mathtt{C},\mathtt{G},\mathtt{T}\}$ the standard basis of $\mathbb{R}^{4}$, then
the components $p_{x_{1}x_{2}x_{3}x_{4}}$ of $P$ are its coordinates in
the natural basis in $\otimes ^{4}\mathbb{R}^{4}$ induced by
$\Sigma $. This motivates the following definition.

\begin{defi}
Given a tensor $P$ in $\left (\mathbb{R}^{4}\right )^{\otimes 4}$, we
denote by $(p_{\mathtt{AAAA}}, p_{\mathtt{AAAC}}, \ldots , p_{\mathtt{TTTT}})^{t}$ the
coordinates of $P$ in the basis
$\{\mathtt{A\otimes A\otimes A\otimes A},\mathtt{A\otimes A\otimes A\otimes C},\mathtt{
\ldots },\mathtt{ T\otimes T\otimes T\otimes T}\}$ induced by $\Sigma $. Similarly,
we write
$\overbar {P}=(\bar{p}_{\mathtt{AAAA}}, \bar{p}_{\mathtt{AAAC}}, \ldots ,
\bar{p}_{\mathtt{TTTG}}, \bar{p}_{\mathtt{TTTT}})^{t}$ for the coordinates of
$P$ in the basis
$\{\mathtt{\bar{A}\otimes \bar{A}\otimes \bar{A}\otimes \bar{A}},\mathtt{ \ldots },\mathtt{
\bar{T}\otimes \bar{T}\otimes \bar{T}\otimes \bar{T}}\}$ induced by the
Fourier basis $\overbar {\Sigma }$.
\end{defi}

The relation between the natural coordinates and the Fourier coordinates
of $P$ is:
\begin{eqnarray*}
\overbar {P} = \left (H^{-1}\otimes H^{-1}\otimes H^{-1}\otimes H^{-1}
\right )P = \frac{1}{4^{4}}\left (H\otimes H\otimes H\otimes H\right )P.
\end{eqnarray*}

\begin{rk}%
\label{rk:FourierOrthogonal}
Since $\frac{1}{2}H$ is an orthogonal matrix, so is
$U := \left (\frac{1}{2}H\right )\otimes \left (\frac{1}{2}H\right )
\otimes \left (\frac{1}{2}H\right )\otimes \left (\frac{1}{2}H\right )$.
Therefore,
\begin{align*}
\parallel {\overbar {P} - \overbar {Q}}\parallel ^{2} &= \parallel {
\frac{1}{2^{4}}U{P} - \frac{1}{2^{4}}U{Q}}\parallel ^{2} =
\frac{1}{4^{4}} \parallel P-Q\parallel ^{2}
\end{align*}
and the Euclidean distance between tensors can be computed using the Fourier
coordinates (up to a positive scalar):
$d({P}, {Q}) = 16 ||\overbar {P}-\overbar {Q}||$.
\end{rk}

If one considers the following bijection between $\Sigma $ and the group
$G:(\mathbb{Z}/2\mathbb{Z}\times \mathbb{Z}/2\mathbb{Z},+)$,
\begin{equation*}
\begin{array}{ccc}
\Sigma =\{\mathtt{A},\mathtt{C},\mathtt{G},\mathtt{T}\}& \longleftrightarrow & \mathbb{Z}/2
\mathbb{Z}\times \mathbb{Z}/2\mathbb{Z}
\\
\mathtt{A} & \mapsto & (0,0)
\\
\mathtt{C} & \mapsto & (0,1)
\\
\mathtt{G} & \mapsto & (1,0)
\\
\mathtt{T} & \mapsto & (1,1)
\\
\end{array}
,
\end{equation*}
then the previous change of coordinates can be understood as the discrete
Fourier transform on $G^{4}$. The following result states that the polynomial
parametrization $\psi _{T}$ becomes monomial in the Fourier parameters:

\begin{thm}%
\label{thm_fourier}
\emph{\citep{evans1993}} Let
$P=\psi _{T}(\pi ,\{M_{i}\}_{i\in [5]})$ where $T$ is the tree topology
$A|B$ and $M_{i}$ are $K81$ matrices. If
$m_{\mathtt{A}}^{i}, m_{\mathtt{C}}^{i}, m_{\mathtt{G}}^{i}, m_{\mathtt{T}}^{i}$ are the Fourier
parameters of $M_{i}$, then the Fourier coordinates of $P$ are
\begin{equation*}%
\bar{p}_{x_{1}x_{2}x_{3}x_{4}} =
\begin{cases}
\frac{1}{4^{4}}m_{\mathtt{x_{1}}}^{1}m_{\mathtt{x_{2}}}^{2}m_{\mathtt{x_{3}}}^{3}m_{
\mathtt{x_{4}}}^{4}m_{\sum _{i\in A}\mathtt{x_{i}}}^{5} &\text{if
$\sum _{i\in A}{\mathtt{x_{i}}} = \sum _{j\in B} {\mathtt{x_{j}}}$,}
\\
0 &\text{otherwise,}
\end{cases}
\end{equation*}
where the sum of elements in $\Sigma $ is given by the bijection
$\Sigma \leftrightarrow \mathbb{Z}/2\mathbb{Z}\times \mathbb{Z}/2
\mathbb{Z}$ introduced above.
\end{thm}

\begin{notation}%
\label{nota_param}
From now on, for the JC69 model, we denote by $x_{i}$ the eigenvalue of
$M_{i}$ of multiplicity three different from 1 (see Lemma~\ref{lema:eigenK81}) and we denote by $\varphi _{T}$ the parameterization
of the phylogenetic varieties from Fourier parameters to Fourier coordinates,
\begin{equation*}
\begin{array}{rcl}
\varphi _{T}: \mathbb{R}^{5} & \longrightarrow & \mathbb{R}^{4^{4}}
\\
\mathbf{x}=(x_{1},x_{2},x_{3},x_{4},x_{5}) & \mapsto &\overbar {P}=
\varphi _{T}(x_{1},x_{2},x_{3},x_{4},x_{5}).
\end{array}
\end{equation*}
The parametrization $\varphi _{T}$ can be computed by adapting Theorem~\ref{thm_fourier} to this model.

\end{notation}

\section{The closest stochastic matrix}\label{sec:matrix}

Throughout this section, we will use the following notation. We write
$\mathcal{H}$ for the hyperplane
$\{x_{1}+\ldots +x_{N} = 1\} \subset \mathbb{R}^{N}$ and
$\Delta :=\{(x_{1},\ldots ,x_{N}) \mid \sum _{i} x_{i}=1, x_{i}\geq 0
\}$ for the standard simplex in $\mathbb{R}^{N}$. Given a point
$\mathbf{x}\in \mathbb{R}^{N}$, we denote by
$\proj _{\mathcal{H}}(\mathbf{x})$ its orthogonal projection onto
$\mathcal{H}$.

\begin{defi}
For any matrix $M\in \mathcal{M}_{N}(\mathbb{R})$ we denote by
$\widehat{M}$ its closest stochastic matrix in the Frobenious norm:
\begin{equation*}
\label{stoc_mat}
\widehat{M} = \argmin _{
\substack{\sum _{j} X_{ij} = 1\ \forall i,\\ X_{ij} \geq 0\ \forall (i,j)}}
\parallel M-X\parallel _{F}.
\end{equation*}
Similarly, for any point $\mathbf{x}\in \mathbb{R}^{N}$ we write
$\widehat{\mathbf{x}}$ for its closest point in $\Delta $.
\end{defi}

The problem of finding the nearest stochastic matrix is equivalent to finding
the closest point (in Euclidean norm) in the standard simplex to every
row of the matrix \cite{Simplex}. The uniqueness of $\widehat{v}$, and
consequently of $\widehat{M}$, is guaranteed since both the objective function
and the domain set are convex. The problem of finding the closest point
in the simplex $\Delta $ to a given point has been widely studied and there
exist several algorithms to compute it. We refer the reader to the work by \cite{Michelot:simplex} for an algorithm that, given any point
$\mathbf{x}\in \mathbb{R}^{N}$, produces the point
$\widehat{\mathbf{x}}\in \Delta $ that minimizes
$\parallel \mathbf{x}-\mathbf{y} \parallel _{2}$ for
$\mathbf{y}\in \Delta $.

\begin{figure}
\includegraphics[scale=0.3]{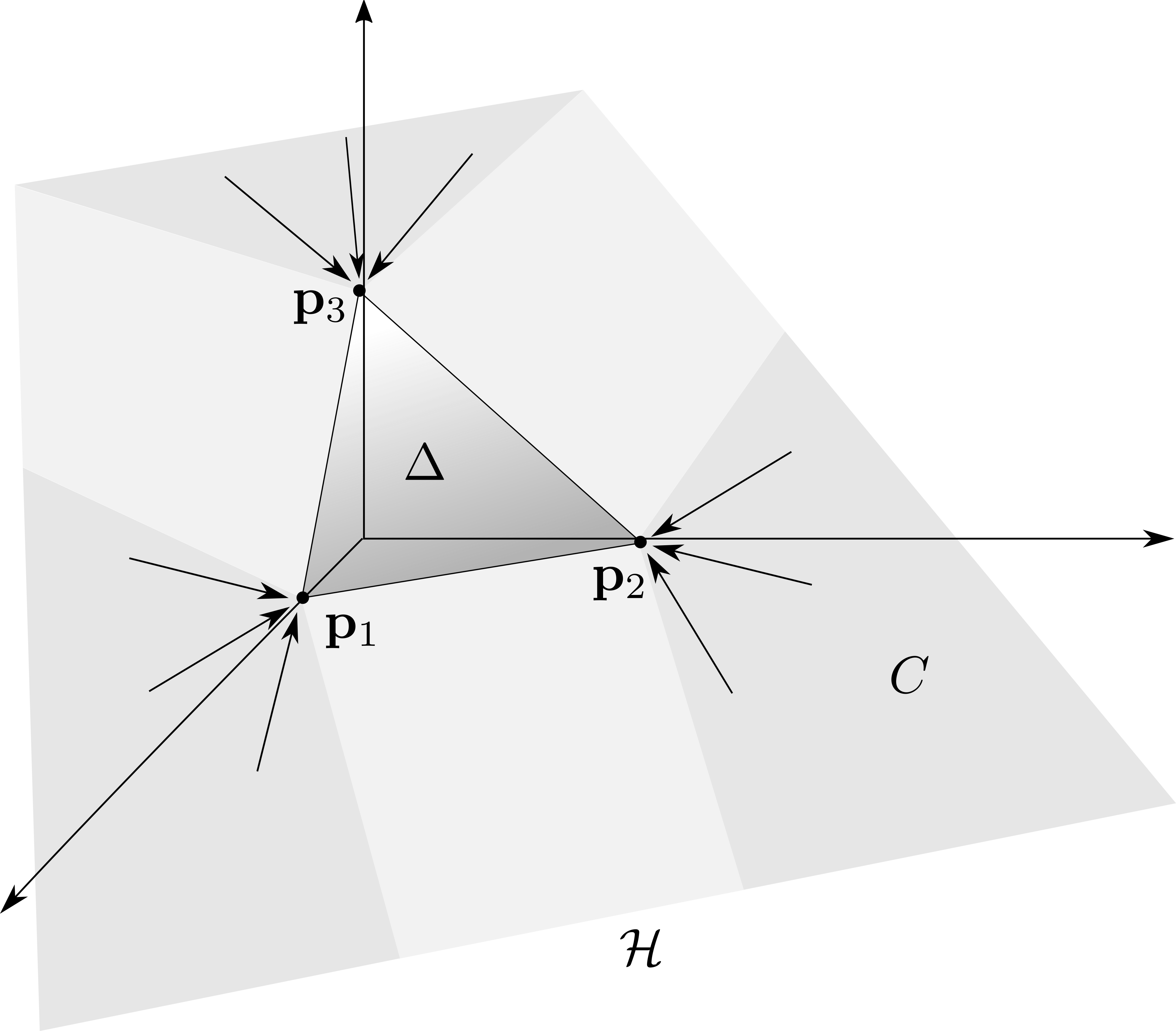}
\caption{\footnotesize{The hyperplane $\mathcal{H}$ and the standard simplex $\Delta $ are
represented in the case $N=3$. The 2-dimensional cone $C$ given by the
inequalities $x_{1}-x_{2}\geq 1$, $x_{1}-x_{3}\geq 1$ corresponds to all the
points in $\mathcal{H}$ whose projection on the simplex is $\mathbf{p}_{1}$ (see
(v) in Lemma~\ref{props_projection}).}}
\label{lba}
\end{figure}

\bigskip
In the following result we state some properties that will be useful later
(see Fig.~\ref{lba} for an illustration of the last item).

\begin{lema}%
\label{props_projection}
Let $\mathbf{x}=(x_{1}, \ldots , x_{N})$ be a point in
$\mathbb{R}^{N}$ and let
$\widehat{\mathbf{x}}=(\widehat{x}_{1}, \ldots , \widehat{x}_{N})$ be its
closest point in $\Delta $.
\begin{enumerate}[(iii)]
\item[\emph{(i)}] $\widehat{\mathbf{x}}$ coincides with the closest point to
$\proj _{\mathcal{H}}(\mathbf{x})$ in $\Delta $,
$\widehat{\proj _{\mathcal{H}}(\mathbf{x}).}$
\item[\emph{(ii)}] If $\mathbf{x}\in \mathcal{H}$ and $x_{i}\leq 0$ for some
$i$, then $\widehat{x}_{i} = 0$.
\item[\emph{(iii)}] Let $\mathbf{y}$ be a point obtained by a permutation of the
coordinates of $\mathbf{x}$, i.e. $\mathbf{y}=P\mathbf{x}$ for some permutation
matrix $P$. Then $\widehat{\mathbf{y}} = P\widehat{\mathbf{x}}$.
\item[\emph{(iv)}] If $x_{i}=x_{j}$ for some $i,j=1,\ldots ,N$ then
$\widehat{x}_{i} = \widehat{x}_{j}$.
\item[\emph{(v)}] $\widehat{\mathbf{x}}$ coincides with
$\mathbf{p}_{i}=(0, \ldots ,\placedown{i}{1},\ldots , 0)$ if and only if
$x_{i} - x_{j} \geq 1$ $\forall j \neq i$.
\end{enumerate}
\end{lema}

\begin{proof}
The proofs of items $(i)$ and $(ii)$ can be found in the paper by \cite{Michelot:simplex} and are the basis of the algorithm provided there.%

\noindent
$(iii)$ It follows from the fact that $P$ is a permutation matrix and hence
is an orthogonal matrix.

\noindent
$(iv)$ This is a direct consequence of $(iii)$.

\noindent
$(v)$ Using $(i)$ and $(ii)$ we can assume that $\sum _{i} x_{i} = 1$,
i.e., $\mathbf{x}$ belongs to the affine hyperplane~$\mathcal{H}$. By
symmetry, it is enough to prove the result for $\mathbf{p}_{1}$, that is,
we prove that $\widehat{\mathbf{x}}=\mathbf{p}_{1}$ if and only if $x_{1}-x_{j}
\geq 1$ for all $j\neq 1$. Firstly we show that if
$x_{1}-x_{j}\geq 1$, $j\neq 1$, then necessarily
$\widehat{\mathbf{x}}=\mathbf{p}_{1}$. Indeed, if
$\mathbf{q}=(c_{1},\ldots ,c_{N})\in \Delta $, we have that
%
\begin{eqnarray}
\label{aux_distances}
d(\mathbf{x},\mathbf{q})^{2} & = & \sum _{j=1}^{N} x_{j}^{2} + \left (
\sum _{j=1}^{N} c_{j}^{2} -2 \sum _{j=1}^{N} c_{j} x_{j} \right )
\\
d(\mathbf{x},\mathbf{p}_{1})^{2} & = & \sum _{j=1}^{N} x_{j}^{2} + (1-2
x_{1}).
\nonumber
\end{eqnarray}
Now, because of the assumption $x_{1}-x_{j}\geq 1$ and
$\sum _{j} c_{j}=1$, we have that
\begin{eqnarray*}
\sum _{j=1}^{N} c_{j} x_{j} \leq \sum _{j=1}^{N} c_{j} x_{1} - \big (
\sum _{j=2}^{N} c_{j} \big ) = x_{1} + (c_{1}-1).
\end{eqnarray*}
In particular,
\begin{eqnarray*}
\sum _{j=1}^{N} c_{j}^{2} -2 \sum _{j=1}^{N} c_{j} x_{j} \geq \sum _{j=1}^{N}
c_{j}^{2} -2 (x_{1}+c_{1}-1) = (c_{1}-1)^{2} + \sum _{j=2}^{N} c_{j}^{2}
+ (1-2\, x_{1}) \geq 1-2\, x_{1}.
\end{eqnarray*}
Comparing this with (\ref{aux_distances}), it follows that
$d(\mathbf{x},q)\geq d(\mathbf{x},\mathbf{p}_{1})$ for any
$q\in \Delta $, so $\mathbf{p}_{1}=\widehat{\mathbf{x}}$. Conversely, assume
that $\mathbf{x}\in \mathcal{H}$ is such that $x_{1}-x_{i}<1$ for some
$i\geq 2$. We will show that there exists some $\mathbf{q}$ in the edge
$\mathbf{p}_{1} \mathbf{p}_{i}$ such that
$d(\mathbf{x},\mathbf{q})<d(\mathbf{x},\mathbf{p}_{1})$ so that
$\mathbf{p}_{1}$ cannot be the closest point to $\mathbf{x}$ in the
simplex $\Delta $. Consider
$\mathbf{q}=a \mathbf{p}_{1} + b \mathbf{p}_{i}$ with $a,b\geq 0$,
$a+b=1$. As above, we have that
\begin{eqnarray*}
d(\mathbf{x},\mathbf{q})^{2} & = & \sum _{j=1}^{N} x_{j}^{2} + a^{2} -2
\,a\,x_{1} +b^{2} -2\,b\,x_{i}.
\end{eqnarray*}

We claim that if we take $0<b<1+x_{i}-x_{1}$, then the point
$\mathbf{q}$ satisfies the inequality between distances above. Indeed,
we need to verify that
\begin{eqnarray*}
a^{2} -2\,a\,x_{1} +b^{2} -2\,b\,x_{i} < 1-2x_{1} = (a+b)^{2} -2 x_{1}
= a^{2}+b^{2} +2ab-2x_{1}.
\end{eqnarray*}

Using that $a=1-b$, this is equivalent to the inequality
$b(b-1+x_{1}-x_{i})<0$, which is satisfied by our choice of $b$.
\end{proof}

\begin{rk}%
\label{example:nonstochasticK81}
If the rows of a matrix $M$ are the result of some permutation applied
to the first row, the previous lemma shows that $\widehat{M}$ will preserve
the same identities between entries as the matrix $M$. Actually, it can
be shown that if $M$ is a matrix in an equivariant model
\citep{Draisma}, then $\widehat{M}$ will remain in the same model. For example,
the matrix
\begin{eqnarray*}
\footnotesize
M = \left (
\begin{array}{c@{\quad }c@{\quad }c@{\quad }c}
0.9 & 0.03 & -0.01 & 0.08
\\
0.03 & 0.9 & 0.08 & -0.01
\\
-0.01 & 0.08 & 0.9 & 0.03
\\
0.08 & -0.01 & 0.03 & 0.9
\\
\end{array}
\right )
\end{eqnarray*}
 is a non-stochastic K81 matrix (see (\ref{eq:K81})). Its nearest
stochastic matrix is
\begin{eqnarray*}
\footnotesize
\widehat{M} = \left (
\begin{array}{c@{\quad }c@{\quad }c@{\quad }c}
{0.89\wideparen{6}} &0.02\wideparen{6} & 0 & 0.07\wideparen{6}
\\
{0.02\wideparen{6}} & {0.89\wideparen{6}} & {0.07\wideparen{6}} & 0
\\
0 & {0.07\wideparen{6}} & {0.89\wideparen{6}} & {0.02\wideparen{6}}
\\
{0.07\wideparen{6}} & 0 & {0.02\wideparen{6}} & {0.89\wideparen{6}}
\\
\end{array}
\right ).
\end{eqnarray*}
 and preserves the same identities between entries. That is,
$\widehat{M}$ remains in the K81 model.
\end{rk}

\begin{lema}%
\label{lem:stochM}
Let $M$ be a $JC69$ matrix. Then $M$ is stochastic if and only if its eigenvalues
lie in $\left [-1/3,1\right ]$.
\end{lema}

\begin{proof}
Let $M$ be a JC69 matrix, that is, $M$ as in (\ref{eq:K81}) with
$c = d = b$, $a=1-3b$. Then, $M$ is stochastic if and only if
$b\geq 0$ and $a = 1-3b \geq 0$, which is equivalent to
$b\in [0,1/3]$. As the eigenvalues of $M$ are $m_{\mathtt{A}} = 1$ and
$m_{\mathtt{C}} = m_{\mathtt{G}} = m_{\mathtt{T}} = 1-4b$ (see Lemma~\ref{lema:eigenK81}), we get that $M$ is stochastic if and only if the
eigenvalue $1-4b$ lies in $[-1/3,1]$.
\end{proof}

\begin{lema}
Let $M$ be a non-stochastic JC69 matrix. Then $\widehat{M}$ is either the
identity matrix or the matrix
\begin{equation*}
\footnotesize%
\begin{pmatrix}

0&1/3&1/3&1/3
\\
1/3&0&1/3&1/3
\\
1/3&1/3&0&1/3
\\
1/3&1/3&1/3&0
\\
\end{pmatrix}%
.
\end{equation*}
\end{lema}

\begin{proof}
Let $M$ be a JC69 matrix with off-diagonal entries equal to $b$ and diagonal
entries equal to $a=1-3b$. Then it is not stochastic if either
$b<0$ or $a<0$. Let $v=(a,b,b,b)$ be the first row of $M$ and
$\widehat{v}=(\widehat{a},\widehat{b},\widehat{b},\widehat{b})$ its projection
onto the simplex $\Delta ^{3}$ (Lemma~\ref{props_projection} $(iv)$). The
following argument is valid for each row due to Lemma~\ref{props_projection} $(iii)$.

If $b<0$ then, by Lemma~\ref{props_projection} $(ii)$, $\widehat{b}$ equals
zero and $\widehat{a}$ has to be equal to $1$ since the coordinates of
$\widehat{v}$ sum to $1$. Therefore $\widehat{M}$ is the $4\times 4$ identity
matrix.

If $a < 0$ then $\widehat{a}=0$ and since $3\widehat{b}=1$,
$\widehat{b} = \frac{1}{3}$. Therefore $\widehat{M}$ is a matrix with
$0$ in the diagonal and $\frac{1}{3}$ at the non-diagonal entries.
\end{proof}

For later use, we close this section by stating a characterization of those
K81 matrices $M$ for which $\widehat{M}$ is a permutation matrix.

\begin{lema}%
\label{lema:perm_matrix}
Let $M$ be a $K81$ matrix and denote by
$(a_{1}, a_{2}, a_{3}, a_{4})$ its first row. Then $\widehat{M}$ is a permutation
matrix if and only if there is some $i\in \{1,\ldots ,4\}$ such that
\begin{equation*}
\label{eq:permutationMatrix}
a_{i} - a_{j} \geq 1 \mbox{ for all } j\neq i.
\end{equation*}
\end{lema}

\begin{proof}
This is an immediate consequence of Lemma~\ref{props_projection}
$(v)$.
\end{proof}

\section{The case of short external branches}\label{sec4}

In this section we study evolutionary processes where substitutions at
the external edges are unusual, so that probabilities of substitution of
nucleotides in the corresponding transition matrices are small. This translates
to matrices close to the identity at the external edges and
\emph{short branch lengths}, as explained in the Introduction.

We use the results of Section~\ref{sec:matrix} with $N=4^{4}$ and we stick
to the K81 model. Given $P\in \mathbb{R}^{4^{4}}$, let $P_{T}^{+}$ be
a point in $\mathcal{V}_{T}^{+}$ that minimizes the distance to $P$, i.e.
\begin{equation*}
d(P,P_{T}^{+})=d(P,\mathcal{V}_{T}^{+}).
\end{equation*}

The following result shows that a point arising from a tree $T$ with no
substitutions at the external edges is always closer to the stochastic
region of $T$ than to any other tree. See Fig.~\ref{fig:prop41b} for
an illustration of this result.

\begin{prop}%
\label{k3_id_prop}
Assume that $P=\psi _{T}(Id, Id, Id, Id, M)$ where $M$ is a non-stochastic
$K81$ matrix and $T$ is any $4$-leaved tree. Then,
\def\theenumi{\alph{enumi}}\def\labelenumi{\emph{(\theenumi)}}
\begin{enumerate}[(a)]
\item The point $P^{+}_{T}$ is equal to
$\psi _{T}(Id, Id, Id, Id, \widehat{M})$. Moreover, $P^{+}_{T}$ coincides
with the point that minimizes the distance to the standard simplex
$\Delta \subset \mathbb{R}^{4^{4}}$. In particular, the point
$P^{+}_{T}$ is unique.
\begin{figure}
\includegraphics[scale=1.8]{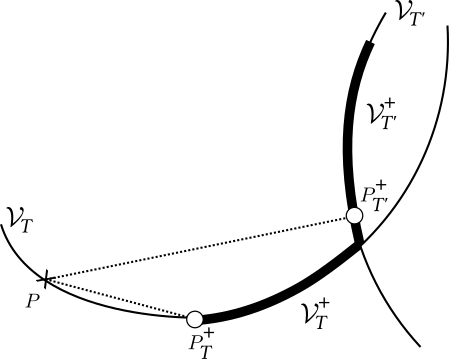}
\caption{\footnotesize{For two different quartet trees $T,T'$, the phylogenetic varieties
$\mathcal{V}_{T}$ and $\mathcal{V}_{T'}$ are represented as curves, with the
intersection being reduced to only one point. The stochastic regions are
represented with thick stroke. The point $P=\psi (Id,Id,Id,Id,M)$ with $M$ not
stochastic lies in $\mathcal{V}_{T}$ but not in the stochastic region
$\mathcal{V}_{T}^{+}$. The points $P_{T}^{+}$ and $P_{T'}^{+}$ represent points
that minimize the distance from $P$ to $\mathcal{V}_{T}^{+}$ and
$\mathcal{V}_{T'}^{+}$, respectively. The figure illustrates that
$d(P,\mathcal{V}_{T}^{+})\leq d(P;\mathcal{V}_{T'}^{+})$ (see (b) of Proposition~\ref{k3_id_prop}).}}
\label{fig:prop41b}
\end{figure}
\item If $T'\neq T$ is another tree in $\mathcal{T}$, then
$d(P,\mathcal{V}^{+}_{T'}) \geq d(P,\mathcal{V}^{+}_{T})$.
\item The following are equivalent:
\begin{enumerate}[(iii)]
\item[\emph{(i)}] equality holds in (b);
\item[\emph{(ii)}]
$P^{+}_{T} \in \mathcal{V}^{+}_{T} \cap \mathcal{V}^{+}_{T'}$;
\item[\emph{(iii)}] the matrix $\widehat{M}$ is a permutation matrix.
\end{enumerate}
\end{enumerate}
\end{prop}

\begin{proof}
We assume that $T = T_{12\mid 34}$, but the proof is analogous for the
other trees. We define $\widehat{P}$ to be the closest point to $P$ in
$\Delta $ (which is a convex set), see Lemma~\ref{props_projection}. First
of all, as $\mathcal{V}_{T}^{+}\subset \Delta $, we have that
%
\begin{eqnarray}%
\label{aux_dist}
d(P,\mathcal{V}_{T}^{+}) = \min _{Q\in \mathcal{V}_{T}^{+}}d(P,Q)
\geq \min _{Q\in \Delta }d(P,Q) =d(P,\widehat{P}).
\end{eqnarray}

We now show that $\widehat{P}\in \mathcal{V}^{+}_{T}$. Since the transition
matrices at the exterior edges of $T$ are the identity matrix, the coordinates
of $P$ are
\begin{equation*}
p_{ijkl} =
\begin{cases}
\frac{1}{4}(M)_{ik} &\text{if $i=j$ and $k=l$}
\\
0 &\text{otherwise.}
\end{cases}
\end{equation*}

Since $M$ is a K81 matrix the non-zero coordinates of $P$ only take
$4$ different values. Moreover, because of Lemma~\ref{props_projection} (ii) and (iv), we can write the coordinates of
$\widehat{P}$ as
\begin{equation*}
\widehat{p}_{ijkl} =
\begin{cases}
b_{ik} &\text{if $i=j$ and $k=l$}
\\
0 &\text{otherwise.}
\end{cases}
\end{equation*}
for some values $b_{ik}$ satisfying the identities of a K81 matrix (see
(\ref{eq:K81})). Since $\widehat{P}$ belongs to the simplex, we have that
$\sum _{i,k} b_{ik}=1$. It follows that the matrix
\begin{eqnarray*}
\footnotesize
4\left (
\begin{array}{cccc}
b_{11} & b_{12} & b_{13} & b_{14}
\\
b_{21} & b_{22} & b_{23} & b_{24}
\\
b_{31} & b_{32} & b_{33} & b_{34}
\\
b_{41} & b_{42} & b_{43} & b_{44}
\\
\end{array}
\right )
\end{eqnarray*}
is a K81 stochastic matrix. Actually, this matrix is just
$\widehat{M}$, and so,
$\widehat{P}=\psi _{T}(Id,Id,Id,Id,\widehat{M})$. In particular,
$\widehat{P}\in \mathcal{V}^{+}_{T}$. Since $P^{+}_{T}$ minimizes the
distance from $P$ to the variety $\mathcal{V}^{+}_{T}$, we have
$d(P,\widehat{P})\geq d(P,P^{+}_{T})$. Because of (\ref{aux_dist}), the
equality holds. Moreover, from the uniqueness of the point minimizing the
distance to $\Delta $, it follows that $P^{+}_{T}=\widehat{P}$. This concludes
the proof of (a).

(b) For any tree topology $T'$, we have that
$\mathcal{V}^{+}_{T'}\subset \Delta $. It follows that
$d(P,\widehat{P})\leq d(P,P^{+}_{T'})$. Since
$\widehat{P}=P^{+}_{T}$, we infer that
$d(P,P^{+}_{T})\leq d(P,P^{+}_{T'})$ for any $T'\neq T$.

(c) Now, we proceed to characterize when the equality holds in (b).

\noindent (i) $\Leftrightarrow $ (ii). It is clear that if
$P^{+}_{T}\in \mathcal{V}^{+}_{T'}$, then
$d(P,\mathcal{V}^{+}_{T})=d(P,P^{+}_{T})\geq d(P,\mathcal{V}^{+}_{T'})$.
Together with the inequality in (b), this proves that (ii) implies (i).
Conversely, if the equality holds, then
$d(P,P^{+}_{T'})=d(P,\Delta )$. Because of the uniqueness of the point
that minimizes the distance to $\Delta $, it follows that
$P^{+}_{T'}=\widehat{P}$, and we have already seen that
$\widehat{P}\in \mathcal{V}^{+}_{T}$. Therefore,
$P^{+}_{T}\in \mathcal{V}^{+}_{T} \cap \mathcal{V}^{+}_{T'}$.

\noindent
(ii) $\Leftrightarrow $ (iii). It only remains to see that
$P^{+}_{T'} = P^{+}_{T}$ (i.e.
$P_{T}^{+}\in \mathcal{V}_{T}^{+}\cap \mathcal{V}_{T'}^{+}$) if and only
if $M$ is a permutation matrix. Assume that
$\widehat{P} \in \mathcal{V}^{+}_{T'}$ and let
$F=flatt_{T'}(\widehat{P})$ be the $16\times 16$ ``flattening'' matrix obtained
by rearranging the coordinates of $\widehat{P}$ according to the bipartition
of leaves induced by $T'$. For example, if $T'=13|24$, then
$F_{(i,j)(k,l)}=\widehat{P}_{ikjl}$. Then it is well known that the rank
of $F$ is less than or equal to $4$ \citep[see][]{allman2003} because
$\widehat{P} \in \mathcal{V}_{T'}$. On the other hand, as
$\widehat{P}=\psi _{T}(Id,\ldots ,Id,\widehat{M})$, $flatt_{T'}(P)$ is
a diagonal matrix whose diagonal is formed by the $16$ entries of
$\widehat{M}$ multiplied by a constant \citep[see][]{allman2009}. The only
way this matrix has rank $\leq 4$ is by imposing the vanishing of 12 entries.
Since $M$ is a $K81$ stochastic matrix, $\widehat{M}$ has to be a permutation
matrix. Conversely, if $\widehat{M}$ is a permutation matrix, then the
corresponding point
$\widehat{P}=\psi _{T}(Id,\ldots ,Id,\widehat{M})$ lies in variety
$\mathcal{V}^{+}_{T'}$ for every $T'\in \mathcal{T}$.
\end{proof}

\begin{rk}
Note that $P_{T}^{+}$ coincides with
$\psi _{T}(Id, Id, Id, Id, \widehat{M})$ but also with any tensor obtained
by a label swapping of the parameters \citep{allman2004b}.
\end{rk}

In the following theorem we prove that, for any point $P$ close enough
to the point $P_{0}=\psi _{T}(Id, Id, Id, Id, {M})$ of Proposition~\ref{k3_id_prop}, the same holds: $P$ is closer to the stochastic region
$\mathcal{V}_{T}^{+}$ than to the stochastic region
$\mathcal{V}_{T'}^{+}$ for $T'\neq T$ (see Fig.~\ref{fig_projection} for an illustration). We need to exclude the case
$d(P, \mathcal{V}_{T}^{+}) = d(P, \mathcal{V}_{T'}^{+})$ (case (c) of
Proposition~\ref{k3_id_prop}) if we want strict inequality.
\begin{figure}
\includegraphics[scale=1.8]{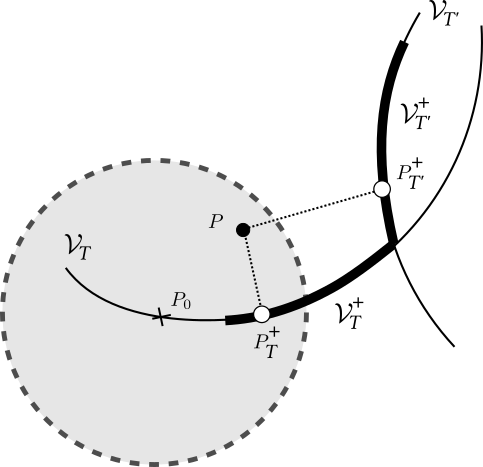}
\caption{\label{fig_projection} \footnotesize{The point $P_{0}$ (see Theorem~\ref{thm:seb}) lies in $\mathcal{V}_{T}$
but not in $\mathcal{V}_{T}^{+}$. As long as a point $P$ lies close to $P_{0}$,
namely $d(P,P_{0}) < (d(P_{0},\mathcal{V}_{T'}^{+}) -
d(P_{0},\mathcal{V}_{T}^{+}))/2$, it will remain closer to the stochastic region
$\mathcal{V}_{T}^{+}$ than to the stochastic region $\mathcal{V}_{T'}^{+}$ for
$T'\neq T$.}}
\end{figure}

\begin{thm}%
\label{thm:seb}
Let $M$ be a $K81$ non-stochastic matrix such that $\widehat{M}$ is not
a permutation matrix (see Lemma~\ref{lema:perm_matrix} for a characterization).
Let $P_{0}=\psi _{T}(Id, Id, Id, Id, M)$,
$T'\in \mathcal{T}\setminus \{T\}$, and let
$P\in \mathbb{R}^{4^{4}}$ be a point such that
\begin{eqnarray*}
d(P,P_{0}) <
\dfrac{d(P_{0},\mathcal{V}_{T'}^{+}) - d(P_{0},\mathcal{V}_{T}^{+})}{2}
\end{eqnarray*}
(this is satisfied if $P$ is close enough to $P_{0}$). Then
$d(P, \mathcal{V}_{T}^{+}) < d(P, \mathcal{V}_{T'}^{+})$.
\end{thm}

\begin{proof}
We first define the function
$f(Q) = d(Q,V_{T'}^{+}) - d(Q,V_{T}^{+})$. By hypothesis,
$\widehat{M}$ is not a permutation matrix and by Proposition~\ref{k3_id_prop}, we have that $f(P_{0}) > 0$. We want to show that
$f(P) > 0$ if $d(P,P_{0})<f(P_{0})/2$. Clearly, we are done if
$f(P)\geq f(P_{0})$, so we assume that $f(P)<f(P_{0})$. From the triangle
inequality we have
$|d(P,\mathcal{W}) - d(P_{0}, \mathcal{W})| \leq d(P,P_{0})$, for any
$\mathcal{W}\subset \mathbb{R}^{N}$. Then, we obtain
\begin{align*}
\lvert f(P) - f(P_{0}) \rvert &= \lvert d(P, \mathcal{V}^{+}_{T'}) - d(P_{0},
\mathcal{V}^{+}_{T'}) - \left (d(P, \mathcal{V}^{+}_{T}) - d(P_{0},
\mathcal{V}^{+}_{T})\right ) \rvert \leq
\nonumber
\\
&\leq \lvert d(P,\mathcal{V}^{+}_{T'}) - d(P_{0}, \mathcal{V}^{+}_{T'})
\rvert + \lvert d(P, \mathcal{V}^{+}_{T}) - d(P_{0}, \mathcal{V}^{+}_{T})
\rvert
\nonumber
\\
& \leq 2 \; d(P,P_{0}) < f(P_{0}).
\end{align*}
Therefore,
$f(P)=(f(P)-f(P_{0}))+f(P_{0})= -|f(P)-f(P_{0})|+f(P_{0}) > 0$. This concludes
the proof.
\end{proof}

\begin{example}
The matrix of Remark~\ref{example:nonstochasticK81} satisfies the hypothesis
of Theorem~\ref{thm:seb}.
\end{example}

\section{Computing the closest point to a stochastic phylogenetic region}
\label{sec5}

Although in the last section we were able to answer our questions analytically, this approach seems unfeasible when we want to tackle more general problems. In this section, in order to find the distance from a point to a stochastic phylogenetic variety we use numerical algebraic geometry.
Our goal is to find all critical points of the distance function to a phylogenetic variety in the interior and at the boundary of the stochastic region. Among the set of critical points we pick the one that minimizes the distance. Similar approaches, where computational and numerical algebraic geometry are applied to phylogenetics studies, can be found in the works \cite{NAG4LS} and \cite{Kosta2019}.

Let $d_\mathcal{X}(x)$ denote the Euclidean distance of a point $x$ to a (complex) algebraic variety $\mathcal{X}$, as a function of $x$.
If $\mathcal{X}_{sing}$ is the singular locus of $\mathcal{X}$, the number of critical points of $d_\mathcal{X}(x)$ in $\mathcal{X}	\setminus \mathcal{X}_{sing}$ for a general $x$ is called the \textit{Euclidean distance degree} (EDdegree for short) of the variety. The EDdegree was introduced in \cite{draisma2015} and it is currently an active field of research. According to Lemma $2.1$ of \cite{draisma2015}, the number of (complex) critical points of $d_\mathcal{X}(x)$ in $\mathcal{X}\setminus \mathcal{X}_{sing}$ is finite and constant on a dense subset.

In this section we assume the JC69 model and we parameterize each transition matrix by its eigenvalue different from $1$ (see Lemma \ref{lema:eigenK81}). As introduced in Section \ref{nota_param}, we denote by $\varphi_{T}(x_1,\ldots, x_5)$ the parameterization in the Fourier coordinates and Fourier parameters for a $4$-leaved tree $T$. Recall that, by Lemma \ref{lem:stochM}, $\varphi_{T}(x_1,\ldots, x_5)$ is a point in the stochastic region if and only if $x_i\in\left[-1/3, 1\right]$, $i=1,\ldots,5$.

Given a point $P$, we denote by $f_T(x_1,\ldots,x_5)$ the square of the Euclidean distance function from the point $\varphi_{T}\left(x_1,\ldots,x_5 \right)$ to $P$:
$$f_T(x_1,\ldots, x_5) = d(P,\varphi_{T}(x_1,\ldots,x_5))^2,$$
and by 
\begin{equation*}\label{eq:D}
	\mathcal{D}:=\left[-1/3,1\right]^5
\end{equation*} 
the region of stochastic parameters.

Under the Jukes-Cantor model, the singular points of the varieties $\VV_T$ are those that are the image of some null parameter. In other words, $\varphi_{T}(x_1,\ldots,x_5)$ is a singular point of the variety if and only if $x_i = 0$ for some $i$ (see \cite{casfer2008} and \cite{casfermich2015} for details).

Hence, we can compute the number of critical points of our function $f_T$ in the pre-image of the smooth part of the variety as the degree of saturation ideal $I:(x_1\cdots x_5)^{\infty}$, where $I$ is generated by the partial derivatives of $f_T$. Using this and the package \texttt{Magma} \cite{magma} we obtain:

\begin{lema}
	If $\VV_{\TT}$ is the phylogenetic variety corresponding to a $4$-leaf tree evolving under the JC69 model, then the EDdegree of $\VV_{\TT}$ is $290$.
\end{lema}

For identifying the critical points of this constrained problem we use the \emph{KKT conditions of first order for local minimums}.

\subsection*{Karush-Kuhn-Tucker conditions (KKT)}\label{KKT}

If $f,g_i:\RR^l \longrightarrow \RR$ are $\mathcal{C^{\infty}}$ functions for $i=1,\dots, n$, we consider the following minimization problem:

\begin{equation*}
	\begin{aligned}
	& \underset{\mathbf{x}}{\text{minimize}}
	& & f(\mathbf{x}) \\
	& \text{subject to}
	& & g_i(\mathbf{x}) \leq 0, \; i = 1, \ldots, n.
	\end{aligned}
\end{equation*}

If a point $\mathbf{x}^*$ that satisfies $g_i(\mathbf{x}^*)\leq0$ $\forall i=1,\ldots,m$ is a local optimum of the problem, then there exist some constants $\mu = (\mu_1,\ldots,\mu_n)$ (called \emph{KKT multipliers}) such that $\mathbf{x}^*$ and $\mu$ satisfy

\begin{itemize}
	\item[(i)] $-\nabla f(\mathbf{x}^*) = \sum_{i=1}^n \mu_i\nabla g_i(\mathbf{x}^*),$
	\item[(ii)] $\mu_i\geq0$ $\forall i=1,\ldots,n$,
	\item[(iii)] $\mu_ig_i(\mathbf{x}^*)=0$ $\forall i=1,\ldots,n$.
\end{itemize}

According to these conditions the algorithm falls naturally into two parts. First of all we find the $290$ critical points of the objective function over all $\CC^5$ and then we check the boundary of $\mathcal{D}$.

To find the critical points at the boundary we restrict the function $f_T$ to all possible boundary subsets and find critical points there. Namely, on the Jukes-Cantor model we write
\begin{eqnarray*}
	g_{1,i}(\mathbf{x}) := x_i -1 \leq 0 \qquad g_{2,i}(\mathbf{x}) := -x_i  - 1/3\leq 0
\end{eqnarray*}
for the inequalities defining the feasible region $\mathcal{D}$. Moreover, for each $i=1,\ldots,5$ and $l=1,2$, write
\begin{eqnarray*}
	S_{l,i}=\{\mathbf{x}=(x_1,\ldots,x_5)\mid g_{l,i}=0\}.
\end{eqnarray*}

Then $x$ is at the boundary of $\mathcal{D}$ if it belongs to the subset $S := \left(\cap_{i\in\iota_1}S_{1,i}\right)\cap\left(\cap_{j\in\iota_2}S_{2,j}\right)$ for some $\iota_1,\iota_2 \subseteq\{1,\ldots,5\}$ disjoint subsets.

We use homotopy continuation methods to solve the different polynomial systems previously described. All computations have been done with the package \texttt{PHCpack.m2} (\cite{PHCpack} and \cite{PHCpackM2}) which turned out to be the only numerical package capable to find these $290$ points of $I:(x_1\cdot \ldots\cdot x_5)^{\infty}.$ \texttt{Macaulay2} \cite{M2} has been used to implement the main core of the algorithm while some previous computations have been previously performed with \texttt{Magma} \cite{magma}. The whole code can be found in \cite{github_marina}.

\IncMargin{1em}
\begin{algorithm}[H]
	\vspace{0.15cm}
	\SetKwInput{KwData}{Input}
	\KwData{A point $P \in \RR^{4^4}$ and a topology $T$.}
	
	\vspace{0.15cm}
	Compute $f_T(\mathbf{x})$\;
	Compute $\mathcal{I}:=\left( \der{f_T}{x_1}, \der{f_T}{x_2}, \der{f_T}{x_3}, \der{f_T}{x_4}, \der{f_T}{x_5} \right)$\;
	$\mathcal{L} := \{\}$ \tcp*[r]{Empty list of valid critical points}
	$d := degree\big(I:(x_1\cdots x_5)^{\infty}\big)$ \;
	Find the $d$ $0$-dimensional solutions of $\nabla f_T$=0 \; 
	\ForEach{solution $x$}{
		\If{$\mathbf{x}\in\RR^5$ and $g_{l,i}(\mathbf{x})\leq0$ $\forall l,i$}{Add $\mathbf{x}$ to $\mathcal{L}$\;}
	}

	\ForEach{ disjoint subsets $\iota_1,\iota_2 \subseteq\{1,\ldots,5\}$}{
		$S := \left(\cap_{i\in\iota_1}S_{1,i}\right)\cap\left(\cap_{j\in\iota_2}S_{2,j}\right)$\;
		Find the solutions of $\nabla (f_T)_{|S}=0$ \;
		\If{$\mathbf{x}\in\RR^5$ and $g_{l,i}(\mathbf{x}) \leq 0$ $\forall l,i$}{
			Add $\mathbf{x}$ to $\mathcal{L}$ \;
		}
	}
	Evaluate each $x\in\mathcal{L}$ into $f_T(x)$ and return the point $x^*$ with minimum $f_T(x^*)$ \;

	\vspace{0.15cm}
	\SetKwInput{KwResult}{Output}
	\KwResult{Parameters $x^*_1,x^*_2,x^*_3,x^*_4,x^*_5$ such that $P_T^+:=\varphi_{T}(x^*_1,\ldots,x^*_5)\in\VV_T^+$ and $d(P,\VV_T^+) = d(P,P_T^+).$}
	\caption{The closest point to a stochastic phylogenetic region \label{algo}}
\end{algorithm}
\vspace{0.2cm}

\section{The long branch attraction case}\label{sec:lba}

\textit{Long branch attraction}, LBA for short, is one the most difficult
problems to cope with phylogenetic inference (see \cite{Kuck2012}). It
is a phenomenon that happens when fast evolving lineages are wrongly inferred
to be closely related. Quartet trees representing these events are characterized
for having two non-sister species that have accumulated many substitutions
and two non-sister species that have very similar DNA sequences.

The length of a branch in a phylogenetic tree represents the expected number
of elapsed mutations along the evolutionary process represented by the
branch and, for the K81 and JC69 models, is estimated as
$-log\big (det(M)\big )/4$, where $M$ is the transition matrix associated
to the edge. Thus, the LBA for quartet trees is represented as in Fig.~\ref{fig:lba}, with two long non-sister branches and two short non-sister
branches and interior edge. As the length of an edge is related to the
eigenvalues of the corresponding transition matrix, for the JC69 model
the eigenvalue different than 1 determines it.

\begin{figure}
\includegraphics[scale=0.7]{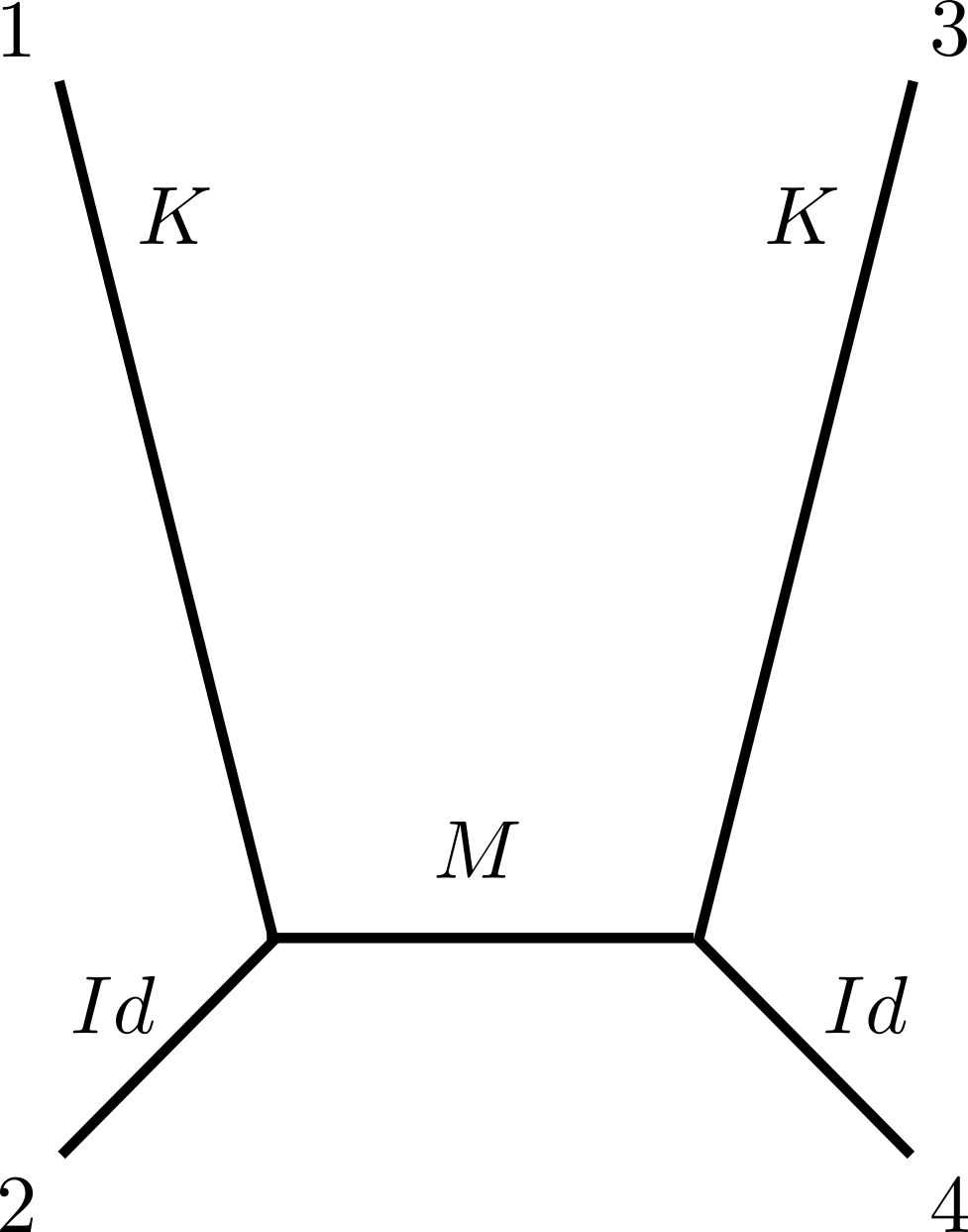}
\caption{Phylogenetic tree $12|34$ with transition matrices leading to the point
$P=\varphi _{12\mid 34}\left (k,1,k,1,m\right )$.}
\label{fig:lba}
\end{figure}

Throughout this section we use the notation introduced in Section~\ref{sec5}. Consider
the tree in Fig.~\ref{fig:lba}, with a non-stochastic matrix $M$ at the
interior edge, a stochastic transition matrix $K$ at edges pointing to
leaves $1$ and $3$, and the identity matrix $Id$ at the remaining edges.
Assume $K$ and $M$ are Jukes-Cantor matrices. Then, let $k$ (respectively
$m$) be the eigenvalue of $K$ (resp. of $M$) different from 1.\footnote{Since
JC69 matrices are determined by their eigenvalue other than 1, we will
adopt the convention of representing the matrix with a capital letter and
its eigenvalue by the same letter in lower case.} Since $K$ is stochastic,
$k$ is in $\left [-1/3, 1\right ]$ (see Lemma~\ref{lem:stochM}). We also
assume $m>1$ since $M$ is not stochastic (the other possibility would be
that $m<-1/3$, but this leads to a biologically unrealistic situation because
in these evolutionary models the transition matrices should not be too
far from the identity matrix if one wants to be able to do inference from
data). Let $P:=\varphi _{12|34}\left (k, 1, k, 1, m\right )$ be the Fourier
coordinates of the corresponding joint distribution.

In this section we study the distance of $P$ to the stochastic phylogenetic
regions $\mathcal{V}_{12|34}^{+}$, $\mathcal{V}_{13|24}^{+}$,
$\mathcal{V}_{14|23}^{+}$ to give an answer to Question 1. As observed
in Remark~\ref{rk:FourierOrthogonal}, we can use Fourier coordinates to
compute distances. Given
$P=\varphi _{12|34}\left (k, 1, k, 1, m\right )$ and
$T\in \mathcal{T}$, we want to find its closest point in
$\mathcal{V}_{T}^{+}$, so our goal is to find
$(x_{1},\ldots ,x_{5})\in \mathcal{D}$ such that
$ d\left (P, \mathcal{V}_{T}^{+}\right ) = d\left (P, \varphi _{T}
\left (x_{1}, x_{2}, x_{3}, x_{4}, x_{5}\right )\right )$.

Therefore, using the notation of Section~\ref{sec5}, we translate the problem of
finding the closest point to
$P=\varphi _{12|34}\left (k, 1, k, 1, m\right )$ in the stochastic phylogenetic
region $\mathcal{V}_{T}^{+}$ can be translated into the following optimization
problem:

\begin{pb}%
\label{OPT}
\begin{equation*}
\begin{aligned}
& \underset{\mathbf{x}}{\text{minimize}} & & f_{T}(\mathbf{x}) := d
\left (P, \varphi _{T}\left (x_{1}, x_{2}, x_{3}, x_{4}, x_{5}\right )
\right )^{2}
\\
& \text{subject to} & & g_{1,i}(\mathbf{x}) \leq 0, \; i = 1, \ldots , 5,
\\
& & & g_{2,i}(\mathbf{x}) \leq 0, \; i = 1, \ldots , 5.
\end{aligned}
\end{equation*}
where $g_{1,i}(\mathbf{x}) = x_{i} -1$ and
$g_{2,i}(\mathbf{x}) = -x_{i} - \frac{1}{3}$.
\end{pb}

\subsection{Local minimum}
\label{sec6.1}

The JC69 phylogenetic variety
$\mathcal{V}_{T} \subset \mathbb{R}^{4^{4}}$ has a linear span
$L_{T}$ of dimension $12$. As the closest point in $\mathcal{V}_{T}$ (resp.
$\mathcal{V}^{+}_{T}$) to a point $P \in \mathbb{R}^{4^{4}}$ coincides
with the closest point to $proj_{L}(P)$ in $\mathcal{V}_{T}$ (resp.
$\mathcal{V}^{+}_{T}$), it is enough to restrict to $L$ to compute optimal
points. However, the varieties $L_{T}$ differ for each tree and their union
spans a linear space of dimension $14$ \citep{casferked}.

For example, for $T=12|34$, the Euclidean distance (in Fourier coordinates)
from $P=\varphi _{12|34} (k, 1, k,\allowbreak  1, m )\in L_{12|34}$ to a point
$\varphi _{12|34}\left (x_{1}, x_{2}, x_{3}, x_{4}, x_{5}\right )\in
\mathcal{V}_{12|34}$ is given by the square root of the following function:
\begin{align*}
f_{12|34}(x_{1}, x_{2}, x_{3}, x_{4}, x_{5})
 :=\ &12\left (x_{1}x_{2}x_{3}x_{4}x_{5}
- k^{2}m\right )^{2} + 9\left (x_{1}x_{2}x_{3}x_{4} - k^{2}\right )^{2}
\\
&+ 6\left (x_{1}x_{2}x_{3}x_{5} - k^{2}m\right )^{2}
+ 6\left (x_{1}x_{2}x_{4}x_{5} - km\right )^{2}
\\
&
+ 6\left (x_{1}x_{3}x_{4}x_{5}
- k^{2}m\right )^{2}
+ 6\left (x_{2}x_{3}x_{4}x_{5} - km\right )^{2}
\\
&
+ 3\left (x_{1}x_{3}x_{5} - k^{2}m\right )^{2} + 3\left (x_{2}x_{3}x_{5}
- km\right )^{2} + 3\left (x_{1}x_{4}x_{5} - km\right )^{2}
\\
&+ 3\left (x_{2}x_{4}x_{5} - m\right )^{2} + 3\left (x_{1}x_{2} - k
\right )^{2} + 3\left (x_{3}x_{4} - k\right )^{2}.
\end{align*}
An initial numerical approach suggests a candidate $\mathbf{x}^{*}$ to
be a minimum of this optimization problem when $T=12|34$ (and also for
the other trees, as we will see later). Define
\begin{equation*}
\omega :=\frac{4}{9} +
\frac{11}{27\sqrt[3]{\frac{69+16\sqrt{3}}{243}}} +
\sqrt[3]{\frac{69+16\sqrt{3}}{243}}\approx 1.734
\end{equation*}
and the intervals
\begin{equation*}
I:=\left [-\frac{1}{3},1\right ]\quad \mbox{ and } \quad \Omega :=
\left (1,\omega \right ].
\end{equation*}
Straightforward computations show that the function
$x\mapsto f_{12|34}(x,1,x,1,1)$ has only one (real) critical point
$\tilde{x}$, when $(k,m)\in I\times \Omega $.

\begin{prop}%
\label{xtilde}
For $(k,m)\in I\times \Omega $ and $T\in \mathcal{T}$, the unique critical
point $\tilde{x}(k,m)$ of the function $f_{T}(x,1,x,1,1)$ is given by the
expression
\begin{eqnarray*}
\tilde{x}(k,m)=\frac{3k^{2}\left (3m + 1\right ) - 4}{36\gamma (k,m)} +
\gamma (k,m),
\end{eqnarray*}
where
$\gamma (k,m) =
\sqrt[3]{\frac{1}{24}k\left (3m + 1\right ) + \frac{1}{216}\sqrt{\alpha (k,m)}}$
and $\alpha (k,m)$ is a positive value given by
\begin{align*}
\alpha (k,m) =&\  -729k^{6}m^{3} - 27k^{6} + 108k^{4} - 243\left (3k^{6} -
4k^{4} - 3k^{2}\right )m^{2} - 63k^{2}
\\
& - 27\left (9k^{6} - 24k^{4} - 2k^{2}
\right )m + 64.
\end{align*}
 Moreover, the function
$\tilde{x}: I\times \Omega \rightarrow \mathbb{R}$ is a continuous function.
\end{prop}

This proposition is proved in Appendix~\ref{app:Thm62}. The computations
in this section and in the Appendix have been done with
\texttt{SageMath} \citep{sagemath} version $8.6$. From now on, we will use
the following notation: given $(k,m)\in I\times \Omega $, we denote
$\mathbf{x}^{*}\in \mathbb{R}^{5}$ the following point:
%
\begin{eqnarray}%
\label{xstar}
\textbf{x}^{*}= \left \{
\begin{array}{c@{\quad }l}
(\tilde{x}(k,m),1,\tilde{x}(k,m),1,1) & \mbox{ if }\tilde{x}(k,m)<1;
\\
(1,1,1,1,1) & \mbox{otherwise.}
\end{array}
\right .
\end{eqnarray}

As the parameter of $\mathbf{x}^{*}$ corresponding to the interior edge
is 1, $\varphi _{T}(\mathbf{x}^{*})$ belongs to the intersection of the
tree phylogenetic varieties
$\mathcal{V}_{12\mid 34}\cap \mathcal{V}_{13\mid 24}\cap
\mathcal{V}_{14\mid 23} $ (see also Lemma~\ref{k3_id_prop}). For that
reason it is natural to ask whether $\mathbf{x}^{*}$ is also a local minimum
of the optimization Problem~\ref{OPT} for $T=13|24$ or $T=14|23$.

\begin{thm}%
\label{thm:LocMin}
If $k\in [-1/3,1]$ and $m\in \Omega $, then $\mathbf{x}^{*}$ is a local
minimum of the optimization Problem~\ref{OPT} for any
$T\in \mathcal{T}$.
\end{thm}

\begin{proof}
In order to prove that $\mathbf{x}^{*}$ is a local minimum we first show
that $\mathbf{x}^{*}$ satisfies the Karush-Kuhn-Tucker (KKT) conditions
defined in Section~\ref{KKT} 
for some KKT multipliers $\mu _{1,i}$,
$\mu _{2,i}$, $i=1,\dots , 5$.

Assume first that $\tilde{x}(k,m) < 1$. Then we observe that
$\partial _{x_{1}} \, f_{12|34} (\mathbf{x}^{*}) = \partial _{x_{3}}
\, f_{12|34} (\mathbf{x}^{*}) = 0$. Moreover we have
$g_{1,i}(\mathbf{x}^{*})=0$ for $i=2,4,5$,
$g_{1,i}(\mathbf{x}^{*})\neq 0$ for $i=1,3$ and
$g_{2,i}(\mathbf{x}^{*})\neq 0$ $\forall i$.

Therefore, by $(iii)$ of the KKT conditions, we need to take
\begin{eqnarray*}
\mu _{2,i} & = & 0, \mbox{ for } i=1,\ldots ,5
\\
\mu _{1,i} & = & 0, \mbox{ for } i = 1,3.
\end{eqnarray*}
Moreover,
$\nabla g_{1,i}(\mathbf{x}) = (0, \ldots , \place{i}{1}, \ldots , 0)^{t}$
for all $i$ and for every $\mathbf{x}$. Therefore condition $(i)$,
\begin{equation*}
-\nabla f_{12|34}(\mathbf{x}^{*}) = \mu _{1,2}\nabla g_{1,2}(
\mathbf{x}^{*}) + \mu _{1,4}\nabla g_{1,4}(\mathbf{x}^{*}) + \mu _{1,5}
\nabla g_{1,5}(\mathbf{x}^{*}),
\end{equation*}
is equivalent to
\begin{equation*}
\left (0, \partial _{x_{2}} \, f_{12|34} (\mathbf{x}^{*}), 0,
\partial _{x_{4}} \, f_{12|34} (\mathbf{x}^{*}), \partial _{x_{5}}
\, f_{12|34} (\mathbf{x}^{*}) \right )^{t} = -(0, \mu _{1,2}, 0,
\mu _{1,4}, \mu _{1,5})^{t},
\end{equation*}
which implies that necessarily
\begin{eqnarray*}
\mu _{1,2} = - \partial _{x_{2}} \, f_{12|34} (\mathbf{x}^{*});
\qquad \mu _{1,4} = -\partial _{x_{4}} \, f_{12|34} (\mathbf{x}^{*});
\qquad \mu _{1,5} = - \partial _{x_{5}} \, f_{12|34} (\mathbf{x}^{*})
\, .
\end{eqnarray*}
Because of condition $(iii)$, to conclude it is enough to show that these
partial derivatives are negative. This is proven in the first part of the
proofs of Lemma~\ref{lema:dif2_negativa}, Lemma~\ref{lema:dif4_negativa} and Lemma~\ref{lema:dif5_negativa} of the Appendix.

As a consequence, the entries of any directional derivative
$\partial _{\mathbf{v}} \, f_{12|34} (\mathbf{x}^{*})$ are less than
or equal to zero for any vector $\mathbf{v}$. Moreover
$\partial _{\mathbf{v}} \, f_{12|34} (\mathbf{x}^{*})$ is the zero vector
if and only if $\mathbf{v}$ belongs to the $x_{1}x_{3}$-plane. As according
to Lemma~\ref{lema:minim} (see Appendix), $\mathbf{x}^{*}$ is a local
minimum if we fix $x_{2}=x_{4}=x_{5}=1$, we can conclude that
$\mathbf{x}^{*}$ is a local minimum of $f_{12|34}$ on $\mathcal{D}$.

If $\tilde{x}(k,m)$ is greater than or equal to $1$, by the KKT conditions
and the same reasoning as before we need to prove that
$\partial _{x_{i}} \, f_{12|34} (\mathbf{x}^{*})$ is negative for every
$i$, since no partial derivative of $f_{12|34}$ vanishes on
$\mathbf{x}^{*}$. This is proven in the second case of Lemmas~\ref{lema:dif2_negativa}, \ref{lema:dif4_negativa},
\ref{lema:dif5_negativa} of the Appendix for $i=2, 4$ and $5$ respectively.
It is a consequence of the second case of Lemma~\ref{lema:minim} that the
partial derivatives with respect to $x_{1}$ and $x_{3}$ are also negatives
(see Corollary~\ref{cor:negx1x3} for the precise statement). Therefore
$\mathbf{x}^{*}$ is a local optimum.

The proof for topologies $13|24$ and $14|23$ follows directly from the
previous results since the functions $f_{13|24}$ and $f_{14|23}$ satisfy
$\partial _{x_{i}} \, f_{13|24} (\mathbf{x}^{*}) = \partial _{x_{i}}
\, f_{14|23} (\mathbf{x}^{*}) =\ \partial _{x_{i}} \, f_{12|34} (
\mathbf{x}^{*})$ for $i\neq 5$ and
$\partial _{x_{5}} \, f_{13|24} (\mathbf{x}^{*})$ and
$\partial _{x_{5}} \, f_{14|23} (\mathbf{x}^{*})$ are also negative by
Lemma~\ref{lema:dif5_1324_negativa} and Lemma~\ref{lema:dif5_1423_negativa} (see Appendix).
\end{proof}

\subsection{Global minimum}
\label{GlobMin}

Although we are not able to prove that the local minimum presented above
is indeed a global minimum, our evidences suggest that it is so for
$T=12|34$:

\begin{conjecture}%
\label{globalP0}
Let $T=12|34$ and
$P_{0}:=\varphi _{T}\left (k_{0}, 1, k_{0}, 1, m_{0}\right )$. If
$(k_{0}, m_{0})\in I\times \Omega $, then
\begin{equation*}
d(P_{0}, \mathcal{V}_{T}^{+}) = d\big (P_{0}, \varphi _{T}\big (
\tilde{x}(k_{0},m_{0}), 1, \tilde{x}(k_{0},m_{0}), 1, 1\big )\big )
\end{equation*}
and
$\varphi _{T}\big (\tilde{x}(k_{0},m_{0}), 1, \tilde{x}(k_{0},m_{0}), 1,
1\big )$ is the unique point in $\mathcal{V}_{T}^{+}$ that minimizes the
distance to $P_{0}$.
\end{conjecture}

\begin{rk}%
We have tested the conjecture for $1000$ pairs of parameters
$(k,m)$ randomly chosen on the region
$\left (0, 1/4\right ] \times \left (1, 3/2\right ]$ in order to simulate
points close to the LBA phenomenon. Every experiment has verified that
the global minimum of the problem is unique and is the point
$\mathbf{x}^{*}$, which is defined as in (\ref{xstar}) (and was proved
to be a local minimum). The computations have been done with
\texttt{Macaulay2} and a list of the tested parameters $k$ and $m$ can
be found in \cite{github_marina}. Though the conjecture is stated for
$T=12|34$, it is also true for any $T\in \mathcal{T}$ by permuting the
parameters accordingly.
\end{rk}

\begin{figure}
\includegraphics[scale=1.8]{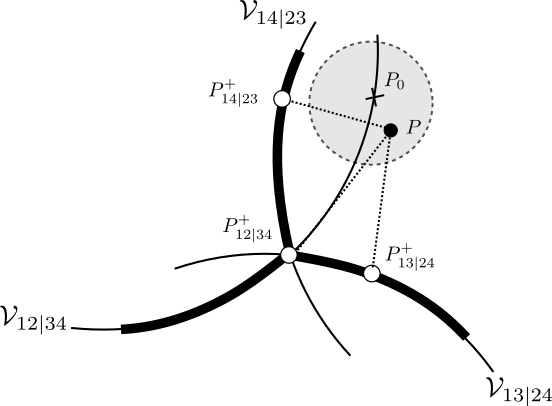}
\caption{The point $P_{0}$ lies in the phylogenetic variety
$\mathcal{V}_{12|34}$ outside the stochastic region ($m>1$). Under the
assumption of Theorem~\ref{thm_lba}, as long as the point $P$ is close to
$P_{0}$, it will remain closer to $\mathcal{V}_{13|24}^{+}$ or
$\mathcal{V}_{14|23}^{+}$ than to $\mathcal{V}_{12|34}^{+}$.}
\label{fig:prop41c}
\end{figure}

In the following theorem, we assume $T=12|34$ and prove that, for any point
$P$ close enough to a point
$P_{0}=\varphi _{T}\left (k_{0}, 1, k_{0}, 1, m_{0}\right )$ satisfying
the previous conjecture, the distance from $P$ to the stochastic phylogenetic
region $\mathcal{V}_{T'}^{+}$, for $T'\neq T$, is upper bounded by the
distance from $P$ to $\mathcal{V}_{T}^{+}$; see
Fig.~\ref{fig:prop41c} for an illustration.

\begin{thm}%
\label{thm_lba}
Let $T=12|34$ and $(k_{0}, m_{0})\in I\times \Omega $, and assume that
$P_{0}:=\varphi _{T}\left (k_{0}, 1, k_{0}, 1, m_{0}\right )\in
\mathcal{V}_{T}$ satisfies that the minimum distance from $P_{0}$ to
$\mathcal{V}_{T}^{+}$ is attained at a unique point $P_{0}^{+}$ given
by
$P_{0}^{+}=\varphi _{T}\big (\tilde{x}(k_{0},m_{0}), 1, \tilde{x}(k_{0},m_{0}),
1, 1\big )\big )$ with $\tilde{x}(k_{0},m_{0}) \neq 0$. Then, if $P$ is close
enough to $P_{0}$ and $T'\neq T$ is another tree in $\mathcal{T}$, its
closest point in $\mathcal{V}_{T}^{+}$ belongs also to
$\mathcal{V}^{+}_{T'}$. In particular,
\begin{equation*}
d(P,\mathcal{V}_{T}^{+}) \geq d(P,\mathcal{V}_{T'}^{+}).
\end{equation*}
\end{thm}

\begin{proof}
We consider the following sets of points in the border of
$\mathcal{V}_{T}^{+}$,
\begin{eqnarray*}
\mathcal{B}_{x_{5}=1} & := & \varphi _{T}(\mathcal{D}\cap \{x_{5}=1\})
\\
\mathcal{B}_{x_{5}=-1/3} &: = & \varphi _{T}(\mathcal{D}\cap \{x_{5}=-1/3
\}).
\end{eqnarray*}
Given a point $P$, we define $f_{T,P}(\textbf{x})$ as the square of the
distance function from $\varphi _{T}(\textbf{x})$ to $P$, and we consider
the set
\begin{equation*}
W_{P} := \left \{  \varphi _{T}(\textbf{x}) \ \bigg \rvert \; \textbf{x}
\in \, \mathcal{D} \mbox{ and } \partial _{x_{5}} \, f_{T,P} (
\textbf{x}) = 0\ \right \}  \cup \mathcal{B}_{x_{5}=-1/3} .
\end{equation*}
Define also $g(P) := d(P, W_{P}) - d(P,\mathcal{B}_{x_{5}=1})$, which is
continuous as a function of $P$.

By hypothesis, $P_{0}^{+}$ equals $\varphi _{T}(\textbf{x}_{0}^{*})$ where
$\textbf{x}_{0}^{*}=\big (\tilde{x}(k_{0},m_{0}), 1, \tilde{x}(k_{0},m_{0}),
1, 1\big )$. Since $\tilde{x}(k_{0},m_{0})\neq 0$,
$\textbf{x}_{0}^{*}$ is the only preimage of $P_{0}^{+}$ \citep[see the work by][]{casfer2008}. Therefore, $P_{0}^{+}$ lies in
$\mathcal{B}_{x_{5}=1}$ but not in $W_{P_{0}}$ (see Lemma~\ref{lema:dif5_negativa}), so that $g(P_{0}) > 0$. If $P$ is close enough
to $P_{0}$, then the function $g(P)$ is still positive. This implies that
the global minimum $P^{+}$ of $f_{T,P}(\textbf{x})$ lies in
$\mathcal{D}$ (which is still unique if $P$ is close to $P_{0}$) does not
lie in $W_{P}$ and therefore lies in the border
$\mathcal{B}_{x_{5}=1}$. As a consequence, $P_{+}$ lies also in
$\mathcal{V}^{+}_{T'}$ and
$d(P,\mathcal{V}_{T}^{+}) \geq d(P,\mathcal{V}_{T'}^{+})$.
\end{proof}

\begin{rk}%
Note that when $k_{0}=1$, this situation is a special case of the situation
considered in Section~\ref{sec4}; the only difference is that we are now restricting
ourselves to the JC69 model instead of considering the K81. The result
obtained here coincides with the case considered in Proposition~\ref{k3_id_prop} (c), where the closest point lies in the intersection
of the varieties and
$d(P, \mathcal{V}_{T}^{+}) = d(P, \mathcal{V}_{T'}^{+})$.
\end{rk}

\section{Study on simulated data}\label{sec:simulations}

In this section we simulate points close to a given phylogenetic variety
and we compute its distance to the stochastic region of this variety as
well as to the other phylogenetic varieties (distinguishing also the stochastic
region of the varieties). We do this in the setting of long branch attraction
of the previous section and for balanced trees. We cannot do this theoretically
because, even if we have found a local minimum for the long branch attraction
setting (Theorem~\ref{thm:LocMin}), we cannot warranty that it is global
and also because we do not have a formula for the distance when the input
does not lie on the variety. The computations of this section are performed
using Algorithm~\ref{algo}.

\begin{figure}
\includegraphics[scale=0.7]{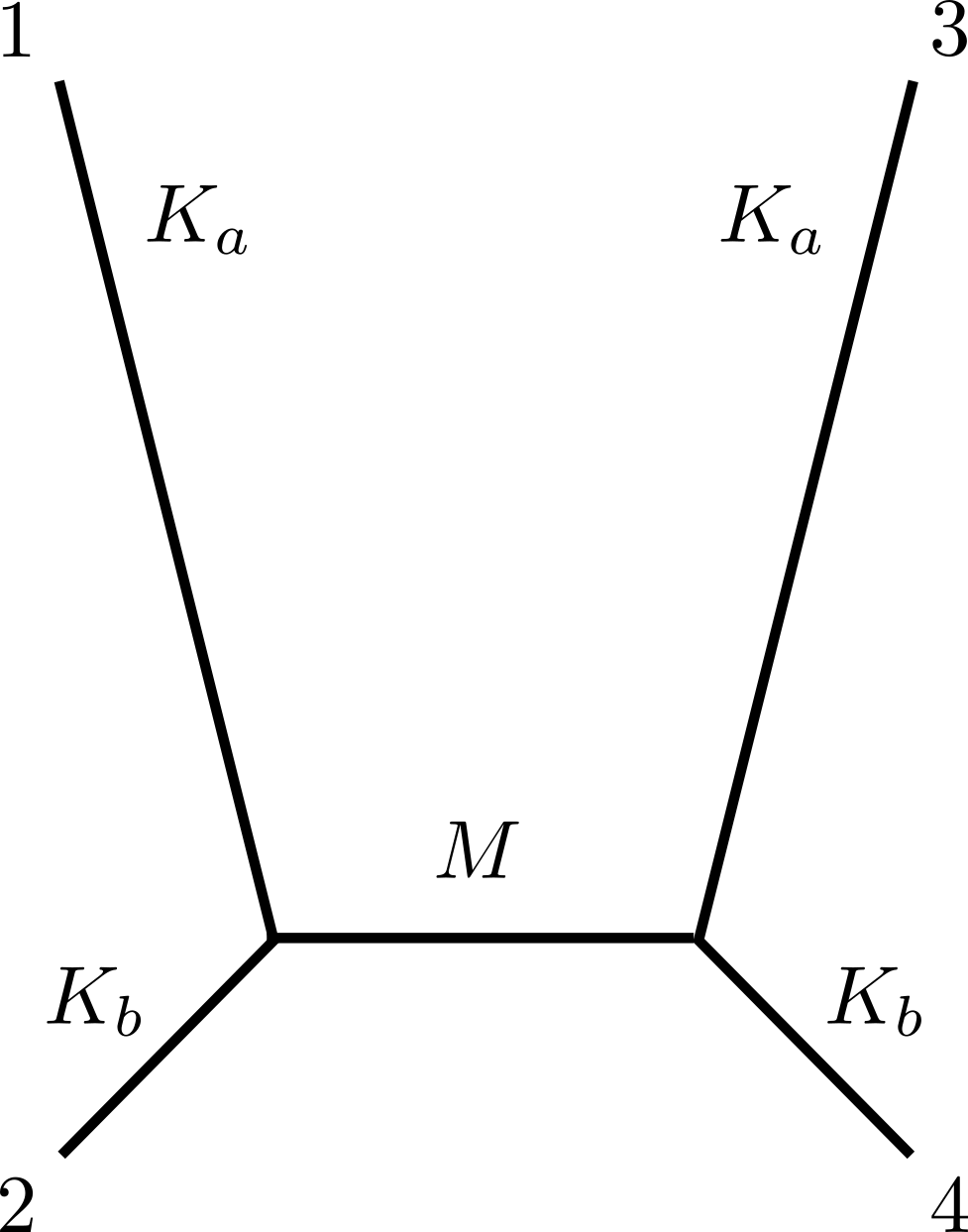}
\caption{Phylogenetic tree with the matrices corresponding to $P = \varphi
(k_{a}, k_{b}, k_{a}, k_{b}, m)$.}
\label{Im:simul_tree}
\end{figure}

We consider a $4$-leaf tree $12|34$ with JC69 matrices. Suppose
$k_{a}$ and $k_{b}$ are the Fourier parameters of matrices at the exterior
edges and $M$ is a JC69 matrix at the interior edge, with eigenvalue
$m$ that takes values in the interval $\left [0.94,1.06\right ]$ (see Fig.~\ref{Im:simul_tree}). These values represent points in
$\mathcal{V}_{12|34}$ that range from the stochastic region of the variety
$\mathcal{V}_{12\mid 34}^{+}$ (that is $m\leq 1$) to the non-stochastic
part ($m>1$). For each set of parameters we considered $100$ data points,
each corresponding to the observation of $10000$ independent samples from
the corresponding multinomial distribution
$\varphi _{T}(k_{a},k_{b},k_{a},k_{b},m)$. As the varieties
$\mathcal{V}_{T}$ all lie in a linear space of dimension $14$ (see the
beginning of Section~\ref{sec6.1}), we first project these data to this linear variety.

For each data point $P$ generated as above and for each tree
$T\in \mathcal{T}$, we have computed the distance of $P$ to the stochastic
region of the variety $\mathcal{V}_{T}^{+}$,
$d(P,\mathcal{V}_{T}^{+})$ using Algorithm~\ref{algo} and we have also
computed the distance to the complete variety,
$ d(P,\mathcal{V}_{T})$. These computations have been performed for the
three tree topologies $12|34$, $13|24$ and $14|23$.

For each set of parameters $k_{a},k_{b}$ and $m$ we have plotted the average
of each of these distances computed from the $100$ data points. In each
graphic we have fixed $k_{a}$ and $k_{b}$ and let $m$ vary in the
$x$-axis from $0.94$ to $1.06$; the $y$-axis represents the distance. The
grey background part of the plots represent the region of data points sampled
from non-stochastic parameters, whereas the white part represents the stochastic
region.

\begin{figure}
\centering
	{Long branch attraction ($k_a = 0.37$, $k_b = 0.87$)}\par\medskip
	\includegraphics[width=14cm]{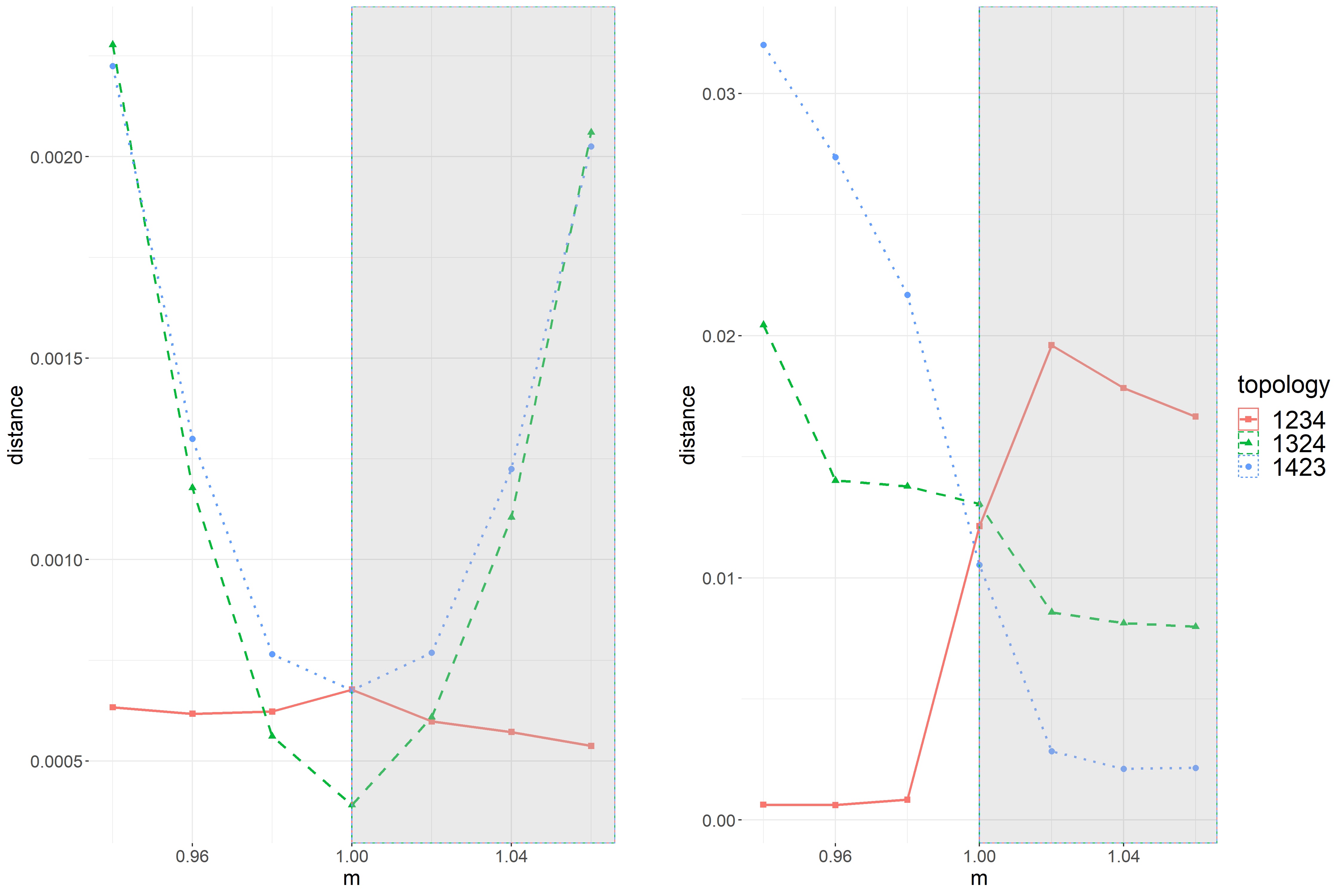}\vspace{0.4cm}
	
	{Balanced trees ($k_a = k_b = 0.51$)}\par\medskip
	\includegraphics[width=14cm]{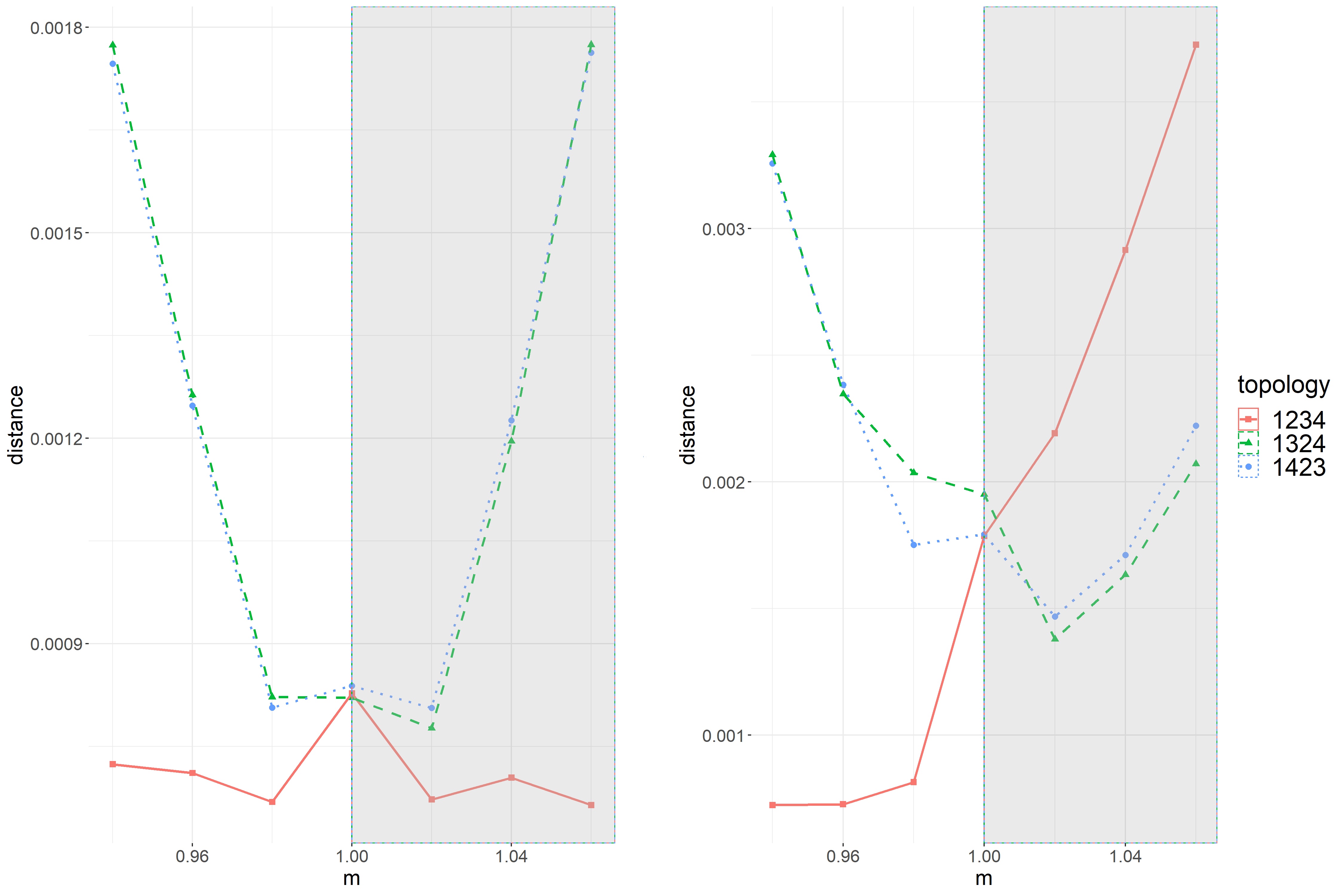}

\caption{These four plots represent the distance of sampled points to the
phylogenetic varieties (on the left) and to their stochastic region (on the
right). In each plot, the horizontal axis represents the eigenvalue $m$ of the
matrix $M$ in the tree of Fig.~\ref{Im:simul_tree}. The two plots on top
correspond to the long branch attraction situation, while the two plots on
bottom correspond to balanced trees. The grey background part indicates the
values of $m$ for which $M$ is not a stochastic matrix.}
\label{Fig:plots_distances}
\end{figure}

The plots on the top of Fig.~\ref{Fig:plots_distances} represent trees
in the long branch attraction (LBA) case (see Fig.~\ref{fig:prop41c}), while those on the bottom represent balanced trees
($k_{a}=k_{b}$); on the left we represent the distance to the phylogenetic
varieties and on the right to the stochastic phylogenetic regions. Concerning
the plots on the left (distance to the phylogenetic varieties), the distance
to $\mathcal{V}_{12|34}$ is always smaller for balanced trees (for all
values of $m$), but this does not hold true in the LBA case (top left figure):
for points close to the intersection of the varieties, that is, $m$ close
$1$, the points are closer to variety corresponding to the tree
$13|23$ (this is the reason why methods based solely on algebraic tools
might perform incorrectly in the LBA case). In both cases (long branch
attraction and balanced trees) we observe a similar behaviour on the plots
on the right (distance to stochastic regions): we note that for
$m\leq 1$ the distance to $\mathcal{V}_{12|34}^{+}$ is almost always the
smallest (except for some points with $m$ very close to 1 in the top figure)
and when $m>1$ the distance to $\mathcal{V}_{12|34}^{+}$ becomes greater
than the distance to the other stochastic regions. This illustrates the
inequality of Theorem~\ref{thm_lba}.

The different performance on the two plots of the distances to
$\mathcal{V}_{13|24}^{+}$ and $\mathcal{V}_{14|23}^{+}$ are due to the
shapes of the trees that we are considering. When the tree is balanced
we see that the distances to $\mathcal{V}_{13|24}^{+}$ and
$\mathcal{V}_{14|23}^{+}$ are almost equal.

Every simulation performed has showed us that, when $m>1$, the closest
point to $P$ in $\mathcal{V}_{12\mid 34}^{+}$, i.e.
$P_{12\mid 34}^{+}$, belongs to the intersection of the varieties, i.e.
$P_{12\mid 34}^{+}\in \mathcal{V}_{12|34}^{+}\cap \mathcal{V}_{13|24}^{+}
\cap \mathcal{V}_{14|23}^{+}$. However, this is not true when we compute
the closest point to $\mathcal{V}_{T'}^{+}$ for $T' \neq 12|34$. In the
case of long branch attraction the closest point
$P_{14|23}^{+}\in \mathcal{V}_{14\mid 23}^{+}$ to $P$ was always the image
of parameters at the interior of $\mathcal{D}$ by $\varphi _{14|23}$ whether
for $T=13 |24$, the parameters describing the closest point to $P$ are
in the interior of $\mathcal{D}$ approximately half of the time.

These simulations verify that, if $P\in \mathbb{R}^{4^{4}}$ is a distribution
satisfying
$d(P,\mathcal{V}_{12|34}) <\min \{d(P,\mathcal{V}_{13|24}),\allowbreak d(P,
\mathcal{V}_{14|23})\}$, it is possible that
$d(P,\mathcal{V}_{12|34}^{+}) > \min \{d(P,\mathcal{V}_{13|24}^{+}),d(P,
\mathcal{V}_{14|23}^{+})\}$. This provides an affirmative answer to the
Question $1$ posed at the beginning of the paper. This suggests that considering
the stochastic part of phylogenetic varieties and the resulting semi-algebraic
constraints needed to describe them may be an interesting strategy for
phylogenetic reconstruction in the long branch attraction setting, and
also for balanced trees. However, as it has become evident throughout this
paper, to deal with both algebraic and semi-algebraic conditions is not
an easy task, and more work is needed in order to design practical methods
for phylogenetic inference under more general evolutionary models than
the models used here.

\subsection{Computations}
\label{sec7.1}

The computations were performed on a machine with $10$ Dual Core Intel(R)
Xeon(R) Silver $64$ Processor $4114$ ($2.20$ GHz, $13.75$ M Cache) equipped
with $256$ GB RAM running Ubuntu $18.04.2$. We have used
\texttt{Macaulay2} version $1.3$ and \texttt{SageMath} version $8.6$.




%
%
%




\section*{Acknowledgements}
The authors would like to thank Piotr Zwiernik for
sharing initial discussion in this topic.
The authors were partially supported by Spanish government Secretar\'{i}a de Estado de Investigaci\'{o}n, Desarrollo e Innovaci\'{o}n [MTM2015-69135-P
(MINECO / FEDER)] and [PID2019-103849GB-I00 (MICINN)]; Generalitat de Catalunya
[2014 SGR-634]. M. Garrote-L\'{o}pez was also funded by Spanish government,
 research project Maria de Maeztu [MDM-2014-0445 (MINECO)]. 
 
\bibliographystyle{plainnat}

\newpage
\appendix
\section{Technical proofs - local minimum}\label{App:TecProofs}

First we recall the notation introduced in Section \ref{sec:lba}. Denote by $\xstar \in \mathbb{R}^5$ the point

\begin{eqnarray*}
	\textbf{x}^*= \left \{ \begin{array}{cl}
		(\tilde{x}(k,m),1,\tilde{x}(k,m),1,1) &  \mbox{ if }\tilde{x}(k,m)<1, \\
		(1,1,1,1,1) & \mbox{otherwise,}
	\end{array}
	\right .
\end{eqnarray*}
where

\begin{equation}\label{eq_tildex}
	\tilde{x}(k,m)=\frac{3k^2\left(3m + 1\right) - 4}{36\gamma(k,m)} + \gamma(k,m),
\end{equation}
$$ \gamma(k,m) = \sqrt[3]{\frac{1}{24}k\left(3m + 1\right) + \frac{1}{216}\sqrt{\alpha(k,m)}},$$ and
\footnotesize
\begin{equation*}
	\alpha(k,m) = -729k^6m^3 - 27k^6 + 108k^4 - 243\left(3k^6 - 4k^4 - 3k^2\right)m^2 - 63k^2 - 27\left(9k^6 - 24k^4 - 2k^2\right)m + 64.
\end{equation*}
\normalsize

Write $\omega =\frac{4}{9} +\frac{11}{27\sqrt[3]{\frac{69+16\sqrt{3}}{243}}} +\sqrt[3]{\frac{69+16\sqrt{3}}{243}}\approx 1.734$ and consider the intervals $I=\left[-\frac{1}{3},1\right]$ and $\Omega = \left(1,\omega\right].$

\subsection{Proof of Proposition \ref{xtilde}}\label{app:Thm62}
In this section we prove the technical results needed to prove Proposition \ref{xtilde}:\\

\noindent\textbf{Proposition \ref{xtilde}.}\textit{
	For $(k,m)\in I\times\Omega$ and $T\in\TT$ the critical point $\tilde{x}(k,m)$ of $f_{T}(x,1,x,1,1)$ is given by the expression \eqref{eq_tildex}. Moreover, $\tilde{x}: I\times \Omega \rightarrow \mathbb{R}$ is a continuous function.
}

\begin{proof}
	Straightforward computations show that $f_{12|34}(x,1,x,1,1)$ = $f_{13|24}(x,1,x,1,1)$ = $f_{14|24}(x,1,x,1,1)$ and that the only real critical point of this function, when $(k,m)\in I\times\Omega$, is the point $$\tilde{x}(k,m)$$ given by expression \eqref{eq_tildex}. In order to prove that $\tilde{x}$ is a continuous real function on $I\times \Omega$, we prove first that $\gamma(k,m)$ is real in Lemma \ref{lema:alphaPos} and then that it does not vanish in Lemma \ref{lema:gamma}.
\end{proof}

\begin{lema}\label{lema:alphaPos}
	$\alpha(k,m)\geq 0$, for all $(k,m)\in I\times\Omega$.
\end{lema}

\begin{proof}
	Consider $\alpha_m(k) := \alpha(k,m)$ as a function of $k$, i.e. suppose m is fixed.
	\footnotesize
	\begin{align*}
		\alpha_m(k) =  \underbrace{\left(-729m^3 -729m^2-243m-27\right)}_{a(m)}k^6 + \underbrace{\left(972m^2+648m+108\right)}_{b(m)}k^4 + \underbrace{\left(729m^2+54m-63\right)}_{c(m)}k^2 + \underbrace{64}_{d}.
	\end{align*}
	\normalsize
	Note that $\alpha(k,m)$ is an even function of $k$ (i.e. $\alpha_m(k) = \alpha_m(-k)$).
	This function has a local minimum at $k = 0$ since
	\begin{equation*}
		\begin{cases}
		\alpha_m(k) = a(m)k^6 + b(m)k^4 + c(m)k^2 + d \mbox{ and } \alpha_m(0) = d = 64 > 0, \\
		\alpha'_m(k) = 6a(m)k^5 + 4b(m)k^3 + 2c(m)k \mbox{ and } \alpha'_m(0) = 0,\\
		\alpha''_m(k) = 30a(m)k^4 + 12b(m)k^2 + 2c(m) \mbox{ and } \alpha''_m(0) = 2\left(729m^2+54m-63\right) > 0 \mbox{ for } m > 1.
		\end{cases}
	\end{equation*}
	
	$\alpha_m(k)$ is an even polynomial of degree $6$ in $k$ with one positive local minimum at $k = 0$. It can be seen that the leading coefficient $a(m)$ is negative for all $m\in \Omega$. It follows that $\alpha_m(k)$ has limit to $-\infty$ when $k$ goes to $\pm\infty$, Thus, its number of real roots will be even and at least two. Suppose $\alpha_m(k)$ has at least $4$ real roots, then the number of local extremes of $\alpha_m(k)$ should be at least seven, but $\alpha'_m(k)$ has degree $5$, and therefore it has at most $5$ roots. Therefore $\alpha_m(k)$ only has $2$ real roots (one positive and one negative) and a (local) minimum at $k = 0$.

	We want to see now that $\alpha_m(k)$ remains positive in $I$ as long as $m\in \Omega$.
	
	Note that for $m=1$, $\alpha_{1}(k) = -1728k^6+1728k^4+720k^2+64$ is zero if and only if $k= \pm\frac{2\sqrt{3}}{3}\approx \pm 1.154\not\in I$.
	On the other hand, the roots of $\alpha_{\omega}(k)$ are $\pm 1$.
	In the following claim we show that, for any $m\in[1,\omega]$, the positive root of $\alpha_m(k)$ is in the interval $\left[1, \frac{2\sqrt{3}}{3}\right]$. By symmetry, the negative root of $\alpha_m(k)$ will be in $\left[-\frac{2\sqrt{3}}{3},-1\right]$ and therefore $\alpha_m(k)$ remains positive in $I$, see Figure \ref{fig:alfamk}.
	
	\begin{figure}[H]
		\centering
		\includegraphics[width=12cm]{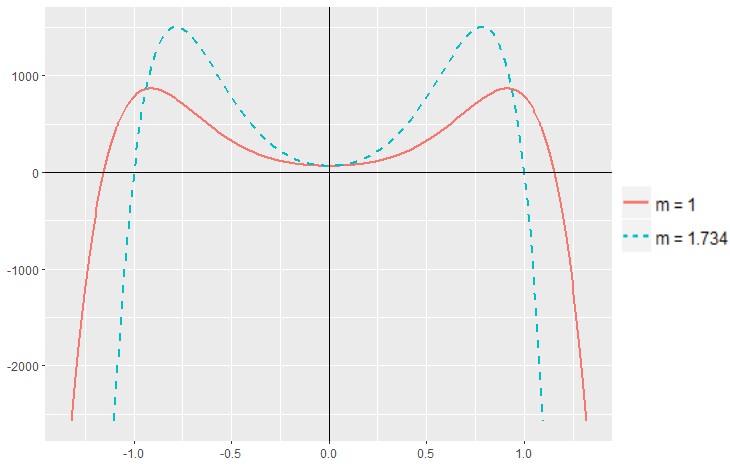}
		\caption{\label{fig:alfamk} The solid line represents $\alpha_{1}(k)$, and the dashed $\alpha_{\omega}(k)$.}
	\end{figure}

	\noindent \emph{Claim.}
	Let $\textbf{k}:[1, \omega] \to \left[1, \frac{2\sqrt{3}}{3}\right]$ be the positive solution of $\alpha_m(k)=0$ (so that $\alpha(\textbf{k}(m), m)=0\ \forall m\in[m_1,m_2]$). Then $\textbf{k}(m)$ is well defined, continuous and strictly decreasing.
	\vspace{2mm}

	\emph{Proof of claim.} As observed above, $\alpha_m(k)$ has exactly one real positive root for any $m>1$. Note that $\textbf{k}(m)$ is continuous by the Implicit Function Theorem. We have seen above that $\textbf{k}(1)=\frac{2 \sqrt{3}}{3}> 1=\textbf{k}(\omega)$. If $\textbf{k}$ was not strictly decreasing, then $\textbf{k}$ would not be injective: there would exist some $m',m''\in[1,\omega]$ such that $\textbf{k}(m') = \textbf{k}(m'')$ and $\alpha(\textbf{k}(m'),m') = \alpha(\textbf{k}(m'),m'') = 0$. %
	In order to reach a contradiction, we show that for any value of $k$, $\alpha(k,m)$  only vanishes for a unique real value of $m$.
	To this aim, consider $\alpha(k,m)$ as a function of $m$,
	\begin{align*}
		\alpha_k(m) =& \underbrace{\left(-729k^6\right)}_{a(k)}m^3 + \underbrace{243\left(-3k^6+4k^4+3k^2\right)}_{b(k)}m^2 + \underbrace{27\left(-9k^6+24k^4+2k^2\right)}_{c(k)}m \\
		& \underbrace{-27k^6+108k^4-63k^2+64}_{d(k)}.
	\end{align*}
	This exhibits $\alpha_k(m)$ as a degree $3$ polynomial in $m$ and it has a unique real root since it has negative discriminant for every $k\neq 0$:
	\begin{align*}
		D(m) &= 18a(k)b(k)c(k)d(k) - 4b(k)^3d(k) + b(k)^2c(k)^2 - 4a(k)c(k)^3 - 27a(k)^2d(k)^2 \nonumber \\
		&= -99179645184(k^6 + 3 k^8) \label{eq:discriminant}
	\end{align*}
	Hence, we conclude that $\textbf{k}(m)$ is well defined and is a strictly decreasing function on $[1,\omega]$.

\end{proof}

\begin{lema}\label{lema:gamma}
	 $\gamma(k,m)\neq 0$, for all $(k,m)\in I\times\Omega$.
\end{lema}	

\begin{proof}
	$\gamma(k,m) = 0 $ if and only if 
	\begin{equation}\label{y0}
		9k\left(3m+1\right) = -\sqrt{\alpha(k,m)}. 
	\end{equation}
	By squaring both members, we derive that $\alpha(k,m) - \big(9k\left(3m+1\right)\big)^2 = 0$. The left member of this expression is equal to $ -\left(9k^2m+3k^2-4\right)^3$, which vanishes if and only if $k = \pm\frac{2}{\sqrt{9m+3}}$. Only the negative solution of $k$ satisfies equation (\ref{y0}). Note that $k = -\frac{2}{\sqrt{9m+3}}$ is always negative and it will be smaller than $-1/3$ if and only if $m < 11/3$.
	
	Therefore, for all $k\in\left[-1/3, 1 \right]$ and $m < \omega < 11/3$, $\gkm$ does not vanish.
\end{proof}

\subsection{Technical results needed for proving Theorem \ref{thm:LocMin}}

In this section we state and prove the results needed to complete the proof of Theorem \ref{thm:LocMin}.\\

To this end, first we need to prove that $\dder{f_{12|34}}{x_2}{\xstar}$ and $\dder{f_{12|34}}{x_4}{\xstar}$ are negative (see Lemmas \ref{lema:dif2_negativa}, \ref{lema:dif4_negativa}) and that $\dder{f_{T}}{x_5}{\xstar}$ is negative for the three topologies $T\in\TT$ (this is done in Lemmas \ref{lema:dif5_negativa}, \ref{lema:dif5_1324_negativa} and \ref{lema:dif5_1423_negativa}). Then we prove that $\xstar$ is a critical point of the function $f_{T}$, for any $T\in\TT$, restricted to the boundary $x_2=x_4=x_5=1$, which is proven in Lemma \ref{lema:minim}. The idea and arguments for the proofs of this section are based on basic concepts and results on \textit{Elimination Theory}. A good general reference for this is Chapter $3$ of \cite{Cox:2007}.

The proofs of these lemmas are divided into two parts. On the first part we assume $\tilde{x}(k,m)<1$ and on the other $\tilde{x}(k,m)$ is assumed to be greater or equal than $1$. For this reason, in the following lemma we start by studying for which parameters $k$ and $m$ one has $\tilde{x}(k,m) \geq 1$.

\begin{lema}\label{lema:m_est}
	It holds that $\tilde{x}(k,m) = 1$ for $(k,m)\in I\times\Omega$ if and only if $m$ is equal to
	\[\mstar :=\frac{-3k^2-k+16}{3k(3k+1)}.\] Moreover, $\tilde{x}(k,m) > 1$ if and only if $m>\mstar$; in this case $k$ is strictly positive.
\end{lema}

\begin{rk}
	It is immediate to check that there are no points $m\in \Omega$ satisfying $\tilde{x}(0,m)=1$ (see Figure \ref{Fig:m_star}). In particular, the condition of the above lemma implies implicitly that the denominator does not vanish.
\end{rk}

\begin{proof}
	Consider new variables $x$, $g$ and $a$ that will allow us to make explicit the algebraic relations of $\tilde{x}(k,m)$, $\gamma(k,m)$ and $\alpha(k,m)$. Then, for $(k,m)\in I \times \Omega$ $\tilde{x}(k,m) = 1$ if and only if $(k,m)$ is a solution of the system of equations:
	
	\begin{equation}\label{eq:defPolys}
		\begin{cases}
		p(x) := x-1 = 0,\\
		p_{\tilde{x}}(x,g,k,m) := 36xg - 36g^2 - 9k^2m - 3k^2 +4 = 0,\\
		p_{\gamma}(g,a,k,m) := 216g^3 - 9k\left(3m + 1\right) - a = 0,\\ 
		p_{\alpha}(a,k,m) := a^2 - \alpha(k,m) = 0.
		\end{cases}
	\end{equation}
	
	Polynomials $p_{\tilde{x}}$, $p_{\gamma}$ and $p_{\alpha}$ stand for the relations introduced in Proposition \ref{xtilde}. Define the ideal $\mathcal{I} := \left(p(x),p_{\tilde{x}}(x,g,k,m), p_{\gamma}(x,g,a,k,m), p_{\alpha}(a,k,m)\right)$ in the polynomial ring $\CC[x,g,a,k,m]$ and compute the elimination ideal $\mathcal{I}\cap\CC[k,m]$. According to Lemma $1$ and Theorem $3$ in section $3.2$ of \cite{Cox:2007}, the variety $\VV(\mathcal{I}\cap\CC[k,m])$ is the smallest algebraic variety containing the possible values $(k,m)$ that correspond to points in $\VV(\mathcal{I})$. However this inclusion is strict and there are points $(k,m)\in\VV(\mathcal{I}\cap\CC[k,m])$ that do not expand to solutions of (\ref{eq:defPolys}).
	
	In this case, the ideal $\mathcal{I}\cap\CC[k,m]$ is generated by the polynomial	
	\begin{equation}\label{eq:j_lemma_mstar}
		\left(9k^2m+3k^2-4\right)^3\left(9k^2m+3k^2+3km+k-16\right).
	\end{equation}
		
	The polynomial vanishes if and only if one of the factors does. The first factor $9k^2m+3k^2-4$ (as a polynomial in $m$) has a root at $m = \frac{4-3k^2}{9k^2}$ and substituting it at (\ref{eq:defPolys}) we get that either
	\begin{equation}\label{eq:noSolsMstar}
	\begin{cases}
		g = 0,\ a = -\frac{12}{k} \mbox{ and } k \neq 0, \mbox{   or}\\
		g = 1,\ a = 108 \mbox{ and }  k = \frac{1}{9}
	\end{cases}
	\end{equation}
	None of these two solutions are satisfied for $k\in I,\ m\in\Omega$. By Lemma \ref{lema:gamma}, $\gamma(k,m)$ is different from zero, then $g$ can not be equal to zero. The second solution in (\ref{eq:noSolsMstar}) implies $m = \frac{107}{3}$, which is not in $\Omega$.
	\\
	The second factor of the polynomial in (\ref{eq:j_lemma_mstar}) vanishes at the points $(k,\mstar)$. By Proposition \ref{xtilde}, $\tilde{x}$ is a continuous real function on $(k,m)$ in $I\times\Omega$. Then to verify when $\tilde{x}(k,m)$ is greater than $1$ it is enough to evaluate it at a point $(k,m)\in I\times\Omega$ such that $m > \mstar$ and at a point $(k,m)\in I\times\Omega$ such that $m < \mstar$. For example, $\tilde{x}(0,3/2) = 0 <1$ and $\tilde{x}(1,3/2) \approx 1.194 > 1$. Therefore, $\tilde{x} > 1$ if and only if $m > m(k)$.
	Straightforward computations show that for any pair $(k,m)\in I\times\Omega$ such that $\tilde{x}(k,m)\geq 1$ it is satisfied that $k>0$ (see Figure \ref{Fig:m_star}).

	\begin{figure}
		\centering
		\includegraphics[width=13cm]{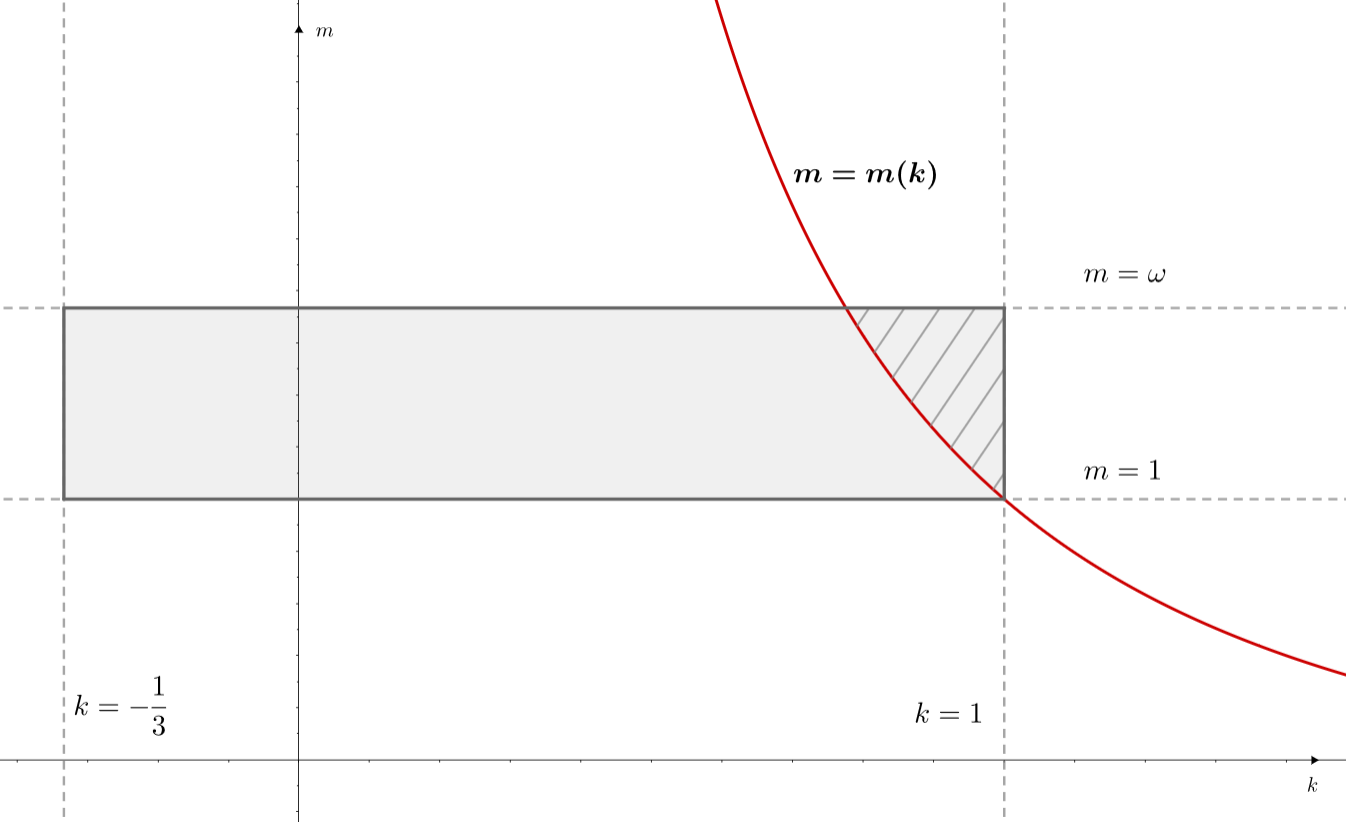}
		\caption{The red curve represents the functions $m = \mstar$ and the grey region is $I\times\Omega$. Therefore the stripped region contains the values $(k,m)\in I\times\Omega$ such that $\tilde{x}(k,m) \ge 1$.
		\label{Fig:m_star}}
	\end{figure}

\end{proof}

The aim of the following three lemmas is to prove that $\dder{f_{12|34}}{x_i}{\xstar}<0$ for $i=2,4,5$. In every lemma, the idea of the proof is the same. We consider an ideal $\mathcal{I}$ for which the contraction in $\mathbb{C}[k,m]$ is the set of points $(k,m)$ such that $\dder{f_{12|34}}{x_i}{\xstar}=0$.
\newpage
\begin{lema}\label{lema:dif2_negativa}
	$\dder{f_{12|34}}{x_2}{\xstar}<0$ for all $(k,m)\in I\times\Omega$. 
\end{lema}

\begin{proof}
	Given $(k,m)\in I\times \Omega$, write $\tilde{x}$ for $\tilde{x}(k,m)$. The proof falls naturally into two cases.
	
	\emph{1st case.} Suppose $\tilde{x}<1$. By definition, $\xstar=(\tilde{x},1,\tilde{x},1,1)$ in this case. Therefore, $\dder{f_{12|34}}{x_2}{(x,1,x,1,1)}$ is given by the polynomial:
	$$p(x,k,m) =  54x^4 -18\left(2k^2m  + k^2 -2\right)x^2 - 6\left(5km + k\right){x}- 6m + 6,\\$$
	To prove that this function is negative we prove that it never vanishes on $I\times\Omega$ and is negative for a particular value in that region. $\dder{f_{12|34}}{x_2}{\xstar}$ is zero if and only if the following polynomials vanish:
		\begin{equation}\label{eq:dif_x2}
		p(x,k,m), \mbox{  }	p_{\tilde{x}}(x,g,k,m),\mbox{  } p_{\gamma}(g,a,k,m), \mbox{ and } p_{\alpha}(a,k,m).
	\end{equation}
	where $p_{\tilde{x}}(x,g,k,m)$, $p_{\gamma}(g,a,k,m)$ and $p_{\alpha}(a,k,m)$ are defined as in (\ref{eq:defPolys}).
	
	We consider the ideal $\mathcal{I} =\left( p(k,m,x), p_{\tilde{x}}(k,m,x,g), p_{\gamma}(k,m,x,g,a), p_{\alpha}(k,m,a)\right)$ and we compute the elimination ideal $\mathcal{I}\cap\CC[k,m]$ which turns out to be generated by exactly one polynomial:
	
	\begin{equation}\label{eq:j_diff2}
		\left(m-1\right)\left(3k^2+1\right)\left(9k^2m+3k^2-4\right)^3h(k,m)
	\end{equation}
	where
\begin{align*}
		h\left(k,m\right) =&\ 81k^6m^3-27k^6m^2-45k^6m-9k^6+39k^4m^3+547k^4m^2,\nonumber \\
		& +469k^4m+97k^4-1312k^2m^2-1120k^2m-256k^2-768m^2.\nonumber
	\end{align*}

	The polynomial in (\ref{eq:j_diff2}) is zero if and only if at least one of its factors vanishes. The first factor is zero when $m=1$, but $1\not\in\Omega$. The second one has no real solutions in $k$. Note that $9k^2m+3k^2-4$ is zero when $k=\pm\frac{2}{\sqrt{9m+3}}$. However, the negative solution does not belong to $I$ if $m\in\Omega$ (see the proof of Lemma \ref{lema:gamma}) and the positive one does not generate a solution of (\ref{eq:dif_x2}). The case of $h(k,m)$ is not that simple. Consider $h$ as a polynomial in $m$:

	\begin{align*}
		h_{k}\left(m\right) =&  \underbrace{\left(81k^6 + 39k^4\right)}_{a(k)}m^3 + \underbrace{\left(-27k^6 + 547k^4 - 1312k^2 - 768\right)}_{b(k)}m^2 +\\ &\underbrace{\left(-45k^6 + 469k^4 - 1120k^2\right)}_{c(k)}m +
		\underbrace{\left(-9k^6 + 97k^4 - 256k^2\right)}_{d(k)}.
	\end{align*}
	The discriminant of $h_{k}\left(m\right)$ is
	\begin{align*}
		D(k) = -49152 k^2 (\sqrt{6} - k) (\sqrt{6} + k) (384 - 106 k^2 + 39 k^4) (64 + 115 k^2 - 38 k^4 + 3 k^6)^2. \\
	\end{align*}
	The discriminant $D(k)$ has three real roots at $k=0$ and $k=\pm\sqrt{6}$ with $\sqrt{6}\sim 2.449$. Since $D(-1) = D(1) < 0$ we conclude $D(k) \leq 0 \ \forall k\in I$ and hence $h_{k}(m)$ only has one real root in this interval. Since the leading coefficient of $h_{k}$ is positive and $h_{k}(2) = 441k^6 + 3535k^4 - 7744k^2 - 3072<0$ we conclude that the root of $h_{k}(m)$ is greater than $2$ and therefore does not belong to $\Omega$.
	
	Consequently there are no points in $\VV(\mathcal{I}\cap\CC[k,m])$ in the region $I\times \Omega$. Since $f_{12|34}$ is continuous and well defined in $I\times\Omega$ it may be concluded that $f_{12|34}$ has the same sign in all the domain. Evaluating at any point $(k,m)\in I\times\Omega$ we conclude that $\dder{f_{12|34}}{x_2}{\xstar}$ is negative on this region.

	\vspace{2mm}
	\emph{2nd case.} Suppose that $\tilde{x} \geq 1$. We already know that in this case, $m\geq \mstar$, which implies that $k>0$ (see figure \ref{Fig:m_star}). On the other hand, we have $\dder{f_{12|34}}{x_2}{\textbf{1}} = -18k^2 - 6(6k^2 + 5k + 1)m - 6k + 96$ is negative if and only if $m > \frac{-3k^2-k+16}{6k^2+5k+1}$. Now, it is straightforward to check that for positive $k$, $\mstar > \frac{-3k^2-k+16}{6k^2+5k+1}$.

\end{proof}

\begin{lema}\label{lema:dif4_negativa}
	$\dder{f_{12|34}}{x_4}{\xstar}<0$ for all $(k,m)\in I\times\Omega$.
\end{lema}

\begin{proof}
	 Computing the partial derivative and substituting we get $\dder{f_{12|34}}{x_4}{\xstar} = \dder{f_{12|34}}{x_2}{\xstar}$. This follows from the symmetry on $f_{12|34}$ and on $\xstar$. Therefore, Lemma \ref{lema:dif4_negativa} is a consequence of Lemma \ref{lema:dif2_negativa}.
\end{proof}

\begin{lema}\label{lema:dif5_negativa}
	$\dder{f_{12|34}}{x_5}{\xstar}<0$ for all $(k,m)\in I\times\Omega$. 
\end{lema}

\begin{proof} 
	We split the proof into two cases

	\emph{1st case.} Suppose $\tilde{x}<1$:		
	\begin{align*}
		\dder{f_{12|34}}{x_5}{\xstar} = 54\tilde{x}^4 - 18(3k^2m - 2)\tilde{x}^2 - 36km\tilde{x} - 6m + 6.
	\end{align*}
	In this case consider the ideal $\mathcal{I} = \left( p(x,k,m),p_{\tilde{x}}(x,g,k,m), p_{\gamma}(x,g,a,k,m), p_{\alpha}(a,k,m) \right)$ where $p(x,k,m) = 54{x}^4 - 18(3k^2m - 2){x}^2 - 36km{x} - 6m + 6$. The ideal $\mathcal{I}\cap\CC[k,m]$ is generated by the polynomial,
	\begin{equation*}\label{eq:j_diff4}
		\left(m-1\right)\left(3k^2+1\right)\left(9k^2m+3k^2-4\right)^3h(k,m)
	\end{equation*}
	where $h(k,m) = 81k^4m^3 - (27k^4+288k^2+256)m^2 - (45k^4+96k^2)m - 9k^4$. We only need to study the intersection of $h(k,m)$ with $I\times\Omega$ since the other factors have already been studied in the proof of Lemma \ref{lema:dif2_negativa}. Taking $h(k,m)$ as a function of $m$ we compute its discriminant,
	\begin{equation*}
		D(k) = -442368 k^6 (2 + 3 k^2) (128 + 18 k^2 + 27 k^4)
	\end{equation*}
	which has only one real root at $k=0$. Substituting at $k=\pm1$ we get $D(-1) = D(1) = -382648320 < 0$. Therefore $D(k) \leq 0 \ \forall k\in I$ and $h_{k}(m)$ has exactly one real root. If $k\in I$ this root is not in $\Omega$ since $h(k,1)=-384k^2 - 256 <0 \ \forall k$, and $h(k,2)=441k^4 - 1344k^2 - 1024 <0 \ \forall k\in I.$
	Same argument as before is valid to conclude $\dder{f_{12|34}}{x_5}{\xstar}$ is negative in our domain.
	\vspace{2mm}

	\emph{2nd case.} Suppose $\tilde{x} \geq 1$: %
	The function $\dder{f_{12|34}}{x_5}{\textbf{1}} =  -6(9k^2 + 6k + 1)m + 96$ is negative if and only if $m > \frac{16}{9k^2+6k+1}$. The value $\mstar$ defined in Lemma \ref{lema:m_est} is greater than $\frac{16}{9k^2+6k+1}$ for all $k\in\left[0, 1\right]$. Since $k>0$ when $\tilde{x}(k,m) > 1$, $\dder{f_{12|34}}{x_2}{\textbf{1}}$ is negative for all $k\in I, m\in\Omega$ such that $m>\mstar$.
\end{proof}

\begin{lema}\label{lema:dif5_1324_negativa}
	$\dder{f_{13|24}}{x_5}{\xstar}\leq 0$ for all $(k,m)\in I\times\Omega$.
\end{lema}

\begin{proof} We split the proof into two cases.
	\emph{1st case.} Assume $\tilde{x}<1$, then:
	\begin{equation*}
		\dder{f_{13|24}}{x_5}{\xstar} = 48\tilde{x}^4 - 12(3k^2m + k^2 - 4)\tilde{x}^2 - 12(3km+k)\tilde{x}.
	\end{equation*}
	Write $p(x,k,m)$ for this polynomial and $\mathcal{I}=\left(p(x,k,m),p_{\tilde{x}}(x,g,k,m), p_{\gamma}(x,g,a,k,m), p_{\alpha}(a,k,m)\right)$. In this case the contraction ideal $\mathcal{I}\cap\CC[k,m]$ is generated by the polynomial
		\[k^4(m-1)(3m+1)^3(9k^2m+3k^2-4)^3\]
	which vanishes if and only if $m =1$, $m=-1/3$, $k = 0$ or $m=\frac{4-3k^2}{9k^2}$. The two first possible values of $m$ do not belong to $\Omega$. If $m=\frac{4-3k^2}{9k^2}$, then $	\dder{f_{13|24}}{x_5}{\xstar}$ vanishes if and only if $k=1/\sqrt{3}$, but then $m=1$, which is not in $\Omega$. It only remains to study the case $k = 0$.
	 Evaluating $\dder{f_{13|24}}{x_5}{\xstar}$ at $k=1$ and $k=-1$, we check that it takes a negative value. Finally, the case $k=0$ implies that $\tilde{x}=0$, which gives $\dder{f_{13|24}}{x_5}{\xstar}=0$.

	\vspace{2mm}
	\emph{2nd case.} Suppose $\tilde{x} \geq 1$:
	The value of $\dder{f_{13|24}}{x_5}{\textbf{1}} = -6(9k^2 + 6k + 1)m + 96$ is negative if and only if $m > \frac{16}{9k^2-6k+1}$. Since the value $\mstar$ obtained in Lemma \ref{lema:m_est} is greater than $\frac{16}{9k^2-6k+1}$ for all $k\in\left[0, 1\right]$ the claim follows.	
\end{proof}

\begin{lema}\label{lema:dif5_1423_negativa}
	$\dder{f_{14|23}}{x_5}{\xstar}\leq 0$ for all $(k,m)\in I\times\Omega$.
\end{lema}

\begin{proof} We split the proof into two cases.\\
	\emph{1st case.} Assume $\tilde{x}<1$, then:
	\begin{equation*}
		\dder{f_{14|23}}{x_5}{\xstar} = 54\tilde{x}^4 - 6(7k^2m + 2k^2 - 6)\tilde{x}^2 - 12(2km+k)\tilde{x} -6m+6
	\end{equation*}
	and write $p(x,k,m)$ for this polynomial. Let $\mathcal{I}:=\left(p(x,k,m),p_{\tilde{x}}(x,g,k,m), p_{\gamma}(x,g,a,k,m), p_{\alpha}(a,k,m)\right)$, then the contraction ideal $\mathcal{I}\cap\CC[k,m]$ is generated by the polynomial
	\begin{equation}\label{eq:h_f14_5}
		(m-1)(9k^2m+3k^2-4)^3h(k,m)
	\end{equation}
		where $h(k,m)= a(m)k^8 + b(m)k^6 + c(m)k^4 + d(m)k^2 + e(m)$ and
	\begin{align*}
		a(m) &= 36m^4 - 129m^3 + 19m^2 + 61m + 13,\\
		b(m) &= -942m^3 + 2362m^2 + 1750m + 286,\\
		c(m) &= -2097m^3 + 7003m^2 + 3853m + 457,\\
		d(m) &= 672m^2 + 8928m + 3072,\\
		e(m) &= 2304m^2
	\end{align*}
	
	The polynomial in (\ref{eq:h_f14_5}) vanishes if $m=1\not\in\Omega$, $k=\pm \frac{2}{\sqrt{9m+3}}$ or $h(k,m)$ is zero. However, recall that $k= -\frac{2}{\sqrt{9m+3}}$ does not belong to $I$ if $m\in \Omega$ (see the proof of Lemma \ref{lema:gamma}) and evaluating $\dder{f_{14|23}}{x_5}{\xstar}$ at $k= \frac{2}{\sqrt{9m+3}}$ one can check that it vanishes if and only if $m=1$ which is not in $\Omega$.

	It remains to see if $h(k,m)$ vanishes for any values $(k,m)\in I\times\Omega$. Straightforward computations show that the roots of the polynomials $a(m)$, $b(m)$, $c(m)$, $d(m)$ and $e(m)$ do not lie in $\Omega$. By evaluating these polynomials at particular values of $\Omega$, it is immediate to check that a(m) is negative, while the other polynomials are positive. Thus, by the Descartes rule $h_m(k)$ (i.e. $h(k,m)$ considered as a function of $k$) has only one positive real root. Since it is an even plynomial on $k$ is has also one real negative root.
	We claim that the positive root of $h_m(k)$ is greater than $1$ for any $m\in\Omega$: observe that $h_m(0) = 2304m^2$ is always positive. Moreover, it is easy to check that the polynomial $h_m(1)= 36m^4 - 3168m^3 + 12360m^2 + 14592m + 3828$ is always positive for any $m\in\Omega$. Then, since $h_m(k)$ has only one positive root and $h_m(0),\ h_m(1) > 0$, the roots of $h_m(k)$ do not lie in $I$ for $m\in\Omega$.
	Evaluating $\dder{f_{13|24}}{x_5}{\xstar}$ at any point in $I\in\Omega$, we check that it takes a negative value.

	\vspace{2mm}
	\emph{2nd case.} Suppose $\tilde{x} \geq 1$:
	The {value of } $\dder{f_{14|23}}{x_5}{\textbf{1}} = -6(7k^2 + 4k + 1)m - 12(k^2+k-8)$ is negative if and only if $m > -\frac{2(k^2+k-8)}{7k^2+4k+1}$. Moreover, $-\frac{2(k^2+k-8)}{7k^2+4k+1} < \mstar$ (see Lemma \ref{lema:m_est} for a definition of $m(k)$) for all $k\in\left[0, 1\right]$. Then, the statement follows.
	
\end{proof}

\begin{lema}\label{lema:minim}
	For any  quartet tree topology $T\in \mathcal{T}$ consider the function  $g:I\times I \longrightarrow \mathbb{R}$ defined as $g(x,y)=f_{T}(x, 1, y, 1, 1)$. Then, the point 
	\begin{eqnarray*}
		\uu:= \left \{ 
		\begin{array}{cl}
			\left(\tilde{x}(k,m),\tilde{x}(k,m)\right) &  \mbox{ if }\tilde{x}(k,m)<1; \\
			(1,1) & \mbox{otherwise;}
		\end{array} \right .
	\end{eqnarray*}
	is a local minimum of $g$.
\end{lema}

\begin{proof}
	Straightforward computations show that $f_{12|34}(x, 1, y, 1, 1) = f_{13|24}(x, 1, y, 1, 1) =  f_{14|23}(x, 1, y, 1, 1) = g(x,y)$. 	Therefore the following proof is valid for any (trivalent) tree topology with 4 leaves. To prove that $\uu$ is a local minimum of $g(x,y)$ we consider two  cases. We first assume that  $\tilde{x}<1$ and we will prove that $\uu$ is a \emph{local} minimum of $g$. The second case is when  $\tilde{x} \geq 1$ so that $\uu$ is on the boundary of $I\times I$. By the KKT conditions we prove that $\nabla g(1,1)$ is negative.
	We write $\tilde{x}$ for $\tilde{x}(k,m)$.
	
	\vspace{2mm}
	\emph{1st case.} Assume $\tilde{x}<1$.
	The first derivatives of $g(x, y)$ vanish at $\uu$. The Hessian matrix of $g$ evaluated at a point $(x,x)$ is
	\begin{equation*}\label{eq:proj}
	\footnotesize
		\Hess = \left(\begin{array}{cc}
 		72 {x}^2 + 24  & -54k^2 m - 18k^2 + 144 {x}^2 \\
 		-54k^2 m - 18 k^2 + 144 {x}^2 & 72 {x}^2 + 24 \\
		\end{array}\right).
	\end{equation*}
 	To show that $\Hess$ is a positive definite matrix, we see that all its principal minors are positive for all $(k,m)\in I\times\Omega$. The first one is clearly positive since it is the sum of positive numbers.
 	To prove that the determinant of $\Hess$ is also positive we will follow the same ideas of the previous lemmas.
	 	
 	Consider the ideal $\mathcal{I} = \big(\det(\Hess),p_{\tilde{x}}(x,g,k,m), p_{\gamma}(x,g,a,k,m), p_{\alpha}(a,k,m)\big)$ where $$\det(\Hess) = -324(3k^2m + k^2 - 8x^2)^2 + 576(3{x}^2 + 1)^2.$$
	 	
 	\noindent The elimination ideal $\mathcal{I}\cap\CC[k,m]$ is generated by the polynomial
 	\begin{equation}\label{eq:j_det}
 		\left(9k^2m + 3k^2 - 4\right)^3\left(27k^2m^2-126k^2m-45k^2-64\right)h(k,m)
 	\end{equation}
	 where
		\begin{align*}
		h(k,m) = &\ 729k^6m^3+729k^6m^2+243k^6m+27k^6-972k^4m^2-648k^4m-108k^4-729k^2m^2 \\
	 	& -54k^2m+63k^2-64.
	\end{align*}
	
	We are interested in the real zeros of each factor of (\ref{eq:j_det}). As in the previous lemmas, it is straightforward to check that the points of the form $(k,\frac{-3k^2 + 4}{9k^2})$ that lie on the domain $I\times \Omega$ do not extend to solutions of the original ideal $\mathcal{I}$; more precisely, $\det(\Hess)$ does not vanish over these points. For any value $k\in I$, the second factor only vanishes at $m = \frac{21k\pm8\sqrt{9k^2+3}}{9k}$, which does not belong to $\Omega$.
	Indeed, if we denote $m^+(k)=\frac{21k + 8\sqrt{9k^2+3}}{9k}$ and $m^-(k)=\frac{21k - 8\sqrt{9k^2+3}}{9k}$, then we want to prove that the image of these functions does not meet $\Omega$. Note that $m^+(k)$ is a decreasing function since its derivative $\dder{m^+}{k}{k} = \frac{-8}{3k^2\sqrt{9k^2+3}}$ is negative for all $k\neq 0$. Moreover,
	\begin{equation*}
		\begin{array}{ll}
			\lim_{k \to -\infty}m^+(k) = \frac{-1}{3}, & \lim_{k \to +\infty}m^+(k) = 5, \\
			\lim_{k \to 0^-}m^+(k) = -\infty, & \lim_{k \to 0^+}m^+(k) = +\infty. \\
		\end{array}
	\end{equation*}
	Hence, $Im(m^+(k)) \cap \Omega$ is empty. The function $m^-(k)$ is increasing since its derivative $\dder{m^-}{k}{k} = \frac{8}{3k^2\sqrt{9k^2+3}}$ is positive for all $k\neq 0$. The limits of this function are
	\begin{equation*}
		\begin{array}{ll}
		\lim_{k \to -\infty}m^-(k) = 5, & \lim_{k \to +\infty}m^-(k) = \frac{-1}{3}, \\
		\lim_{k \to 0^-}m^-(k) = +\infty, & \lim_{k \to 0^+}m^-(k) = -\infty \\
		\end{array}
	\end{equation*} 
	and therefore the image of $m^-(k)$ neither intersects with $\Omega$.
	
	Consider $h(k,m)$ as a function of $m$. As its discriminant $D(k) = -297538935552k^8 - 99179645184k^6$ is negative for all $k\neq 0$, then the polynomial $h_{k}(m)$ has at most one real root $\forall k$.
	Moreover, $h_{k}(1)\leq 0$ and $h_{k}(\omega)\leq 0$ for all $k$ and hence $h_{k}$ is smaller or equal than zero for all $m\in\Omega$.
	Therefore it can be deduced that $\det(\Hess)$ has constant sign in the region $I\times\Omega$. Substituting at a particular point on that region we check that $\det(\Hess) > 0$ for all $(k,m)\in I\times\Omega$.	\\
	 	
	\emph{2nd case.} Assume $\tilde{x} \geq 1$. In this case, since we are in the boundary of the domain, we need to prove that $\nabla g(1,1)<0$.
	 The gradient $$\nabla g(1,1) = (-54k^2m - 18k^2 - 18km - 6k + 96, -54k^2m - 18k^2 - 18km - 6k + 96)$$ is zero if and only if $m = \mstar$. Moreover for $m\geq \mstar$ or equivalently for $\tilde{x}\geq 1$ the polynomial $-54k^2m - 18k^2 - 18km - 6k + 96$ is negative.
\end{proof}

\begin{cor}\label{cor:negx1x3}
	$\dder{f_{T}}{x_1}{\xstar}$ and $\dder{f_{T}}{x_3}{\xstar}$ are less than or equal to zero for any $T$.
\end{cor}
\begin{proof}
	For any $T$, $\dder{f_{T}}{x_1}{\xstar} = \dder{g}{x}{\uu}$ and $\dder{f_{T}}{x_3}{\xstar} = \dder{g}{y}{\uu}$, where $g(x,y)$ and $\uu$ are defined as in the previous lemma. Therefore, as shown in Lemma \ref{lema:minim} the partials $\dder{f_{T}}{x_1}{\xstar}$ and $\dder{f_{T}}{x_3}{\xstar}$ are zero if $\tilde{x}(k,m)<1$ and negative if $\tilde{x}(k,m)\geq1$.
\end{proof}

\end{document}